\documentclass[reqno,11pt]{amsart}

\newtheorem{corollary}{Corollary}
\newtheorem{definition}{Definition}

\usepackage{graphicx}
\usepackage{xcolor}
\usepackage{tikz}
\usepackage{tikz-cd}
\usepackage{caption}
\usepackage{booktabs}
\usepackage{adjustbox}
\usepackage{longtable}
\usepackage{array}
\usepackage{amssymb}

\usepackage{amsmath}

\newcommand{\be}{\begin{equation}}
\newcommand{\ee}{\end{equation}}

\DeclareMathSymbol{\Lambda}{\mathord}{operators}{"03}

\usepackage{enumerate}
\usepackage{enumitem}

\usepackage{arydshln}

\usepackage{appendix}

\newtheorem{thm}{Theorem}[section]
\newtheorem{prop}[thm]{Proposition}
\newtheorem{theorem}[thm]{Theorem}

\newtheorem{lemma}[thm]{Lemma}

\newtheorem{remark}[thm]{\it Remark}

\usepackage{tikz-3dplot}
\usepackage{xifthen}
\usepackage{caption}

\newcommand{\cp}{\mathbb C\mathbb P^1}

\tdplotsetmaincoords{60}{125}

\usepackage[
labelfont=sf,
hypcap=false,
format=hang
]{caption}

\usepackage{mathtools}

\begin{document}

\title{Multicomponent pentagon maps}

\author[P. Kassotakis]{Pavlos Kassotakis}\footnote{Corresponding author: P. Kassotakis}
\address{Pavlos Kassotakis, Department of Mathematics, University of Patras, 26 504 Patras, Greece}
 \email{pkassotakis@math.upatras.gr}


\date{\today}

\begin{abstract}
We provide necessary and sufficient conditions for maps that satisfy  associative-like conditions on families of n-ary magmas to be pentagon maps. We obtain parametric-pentagon maps and we propose a procedure that generates families of multicomponent pentagon and entwining pentagon maps from a given  pentagon map.    
\end{abstract}

\maketitle

\section{Introduction}
The pentagon equation reads
\begin{align}\label{pentagon}
  S_{12}S_{13}S_{23}= & S_{23}S_{12},
\end{align}
where  $S$ may denote a linear operator or a map. When $S$ is a linear operator, we have the operator version of (\ref{pentagon}) and the indices denote the vector spaces where there is non-trivial action. When $S$ is a map, we have the set-theoretic version of (\ref{pentagon}) and the indices denote the components of a threefold cartesian product on which the map acts nontrivially.  Solutions of the set-theoretic version are referred to as set-theoretical solutions of the pentagon equation or simply pentagon maps.

 The pentagon equation first appeared in the theory of angular momentum as an identity satisfied by the Racah coefficients \cite{Biedenharn:1953,Elliott:1953}.   
Over the years, the pentagon equation  has been shown to be related to several mathematical structures, including quasi-Hopf algebras \cite{Drinfeld_p}, conformal field theory  \cite{MoorSeib:89}, geometric topology \cite{Korepanov:2000}, incidence geometry \cite{Doliwa:2014p}, functional analysis \cite{Faddeev:1994}, and  integrable systems \cite{Maillet:1990,Kaufmann_1993,Doliwa:2020}. Systematic studies on set-theoretical solutions of the pentagon equation  were  carried out in \cite{Kashaev:1998}. For a survey of set-theoretical solutions we refer the reader to \cite{Mazzotta:2023}, as well as to \cite{Mazzotta_2023,Mazzotta2025} for various types of set-theoretic solutions. Additional developments, including the study of combinatorial structures underlying the pentagon equation, can be found in \cite{Dimakis:2015,Hoissen2024}, while for developments on linear problems associated to pentagon maps we refer to \cite{Hoissen2024,Kassotakis:2023_p}. For classification results see \cite{Catino:2020,Colazzo:2020,Kassotakis_2024}, and for set-theoretic solutions of higher simplex equations which arise from  pentagon maps see   \cite{Kashaev:1998,Dimakis_Korepanov:2021,KASSOTAKIS224,Mihalache2026}.

In this article we focus on the set-theoretic version of the pentagon equation. Nevertheless, many of the subsequent results can be translated to the operator version of  (\ref{pentagon}). Moreover, all of the results  can be formulated in terms of  the reverse pentagon equation or the braid pentagon equations which  are equivalent  to (\ref{pentagon}) and are given respectively by
\begin{align*}
S_{23}S_{13}S_{12}= & S_{12}S_{23},\\
S_{12}S_{23}S_{12}= & S_{23}\tau_{12}S_{23},\\
S_{23}S_{12}S_{23}= & S_{12}\tau_{23}S_{12},
\end{align*}
where $\tau$ permutes the components, which are indicated by the indices, of the threefold cartesian product.

In Section \ref{sec2}, we revisit the connection of pentagon maps with incidence geometry and we provide necessary and sufficient conditions for maps which are equivalent to  associative-like conditions on families of partial  magmas to be pentagon maps. As an example, from an affine family of binary operations we recover a family of pentagon maps parameterized by $\alpha>0\in\mathbb{R}$ which was introduced in \cite{Kashaev:1998}. Furthermore, we prove the Liouville integrability of this family of pentagon maps and we present associated  families of tetrahedron and hexagon maps. 
In Section \ref{sec3_0} we generalize the results to the case of $n-$ary partial magmas. That is we provide necessary and sufficient conditions for maps which are equivalent to  associative-like conditions on families of $n-$ary partial magmas to be pentagon maps. In the ternary case, as an example,  we generalize the example of  Section \ref{sec2} and we obtain two families of two-component pentagon maps. Also we propose the notion of {\em parametric pentagon maps} and we provide explicitly families of  such pentagon maps. Finally in Section \ref{sec3}, we propose  a construction that produces families of multicomponent maps from a single given map. When the  given map is a pentagon map, we obtain multicomponent pentagon and entwining pentagon maps.

\section{Binary operations and  pentagon maps}\label{sec2}
Let $\mathcal{X}$ be a set. We have the following definitions.

\begin{definition}
A map $S:\mathcal{X}\times \mathcal{X}\rightarrow \mathcal{X}\times \mathcal{X}$ is called   {\em pentagon map}
if it is a solution to the set-theoretic version of the  pentagon equation (\ref{pentagon}), where  $S_{ij},$ $i< j\in \{1,2,3\},$ denote the maps that act as $S$ on the $i^{th},$ and on the $j^{th},$ components of $\mathcal{X}\times\mathcal{X}\times\mathcal{X}$ and as identity to the remaining one.
\end{definition}

\begin{definition}\label{def0}
We say that a  map $S:\mathcal{X}\times \mathcal{X}\rightarrow \mathcal{X}\times \mathcal{X},$ $S:(x,y)\mapsto (u,v),$ is equivalent to the associativity condition
\begin{align}\label{associativity22}
\left<a\left<bc\right>_y\right>_x=& \left<\left<ab\right>_uc\right>_v,
\end{align}
for a family of binary operations $\left<\cdot\right>_x:\mathcal{I}\times \mathcal{I}\rightarrow \mathcal{I},$ parameterized by $x\in\mathcal{X}$, if  the associativity condition (\ref{associativity22}) implies the map $S$.
\end{definition}

In \cite{Hoissen2024} it was shown that some pentagon maps are equivalent  to the associativity condition of some specific binary operations. For example, it was shown that the pentagon map $S_I:\mathbb{CP}^1\times \mathbb{CP}^1\rightarrow \mathbb{CP}^1\times \mathbb{CP}^1$ that reads
\begin{align*}
S_I: (x,y)\mapsto (u,v)=\left(\frac{x}{x+y-xy},x+y-xy\right),
\end{align*}    
is equivalent to the associativity condition (\ref{associativity22})
where the family of binary operations  is defined by 
\begin{align}\label{al_quandle}
\left<\cdot\right>_x: (a,b)\mapsto & \left<ab\right>_x:= x a+(1-x) b.
\end{align}
The binary operation represents the
collinearity of three points $a$, $b$, $\left<ab\right>_x$ and the associativity condition serves as a consistency relation on the Menelaus configuration, see Figure  \ref{Menelaus}. Then the pentagon equation reads as a consistency condition on the  Desargues configuration $(10_3)$ that contains five Menelaus configurations, see Figure \ref{des} (cf. \cite{Doliwa:2014p}). 
%



%
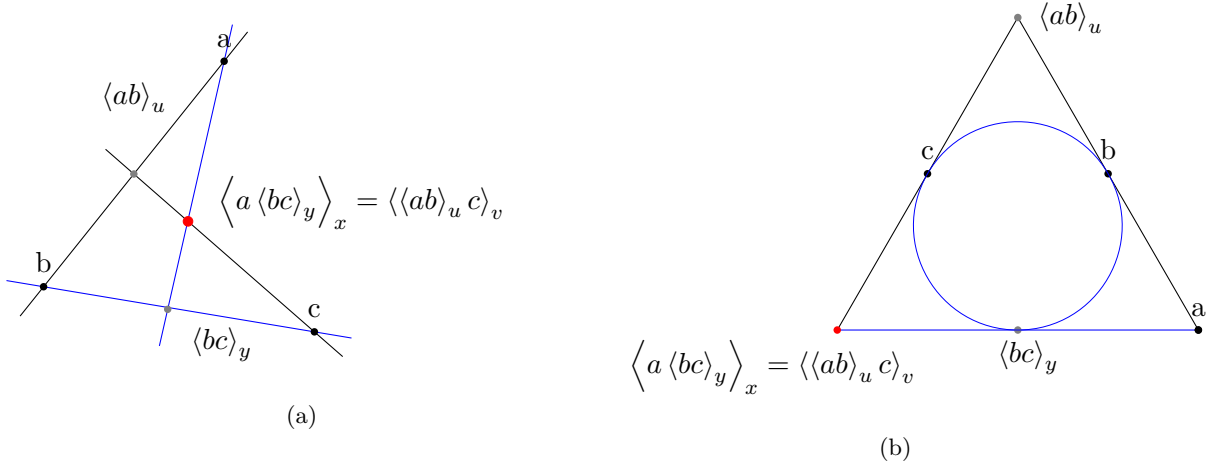
\begin{figure}
\begin{minipage}[htb]{0.45\textwidth}
\begin{tikzpicture}[dot/.style={circle,inner sep=1pt,fill,label={#1},name=#1},
  extended line/.style={shorten >=-#1,shorten <=-#1},
  extended line/.default=1cm,scale=0.6]

\coordinate (aa) at (0,0);  
\coordinate (bb) at (2.75,-0.5);
\coordinate (cc) at (6,-1); 
\coordinate (d) at (2,2.5);
\coordinate (e) at (3,1);
\coordinate (f) at (4,5); 

\draw[extended line=0.5cm,blue] (aa)--(cc); 
\draw[extended line=0.5cm] (aa)--(f); 
\draw[extended line=0.5cm] (cc)--(d); 
\draw[extended line=0.5cm,blue] (f)--(bb);

\node [dot=b] at (aa) {};
\node [dot=a] at (f) {};
\node [dot=c] at (cc) {};
\node [label={[label distance=0.55cm]90:$\left<ab\right>_u$}]  at (d) {}; \filldraw[gray] (d) circle (2pt);
\node [label={[label distance=0.05cm]330:$\left<bc\right>_y$}]  at (bb) {}; \filldraw[gray] (bb) circle (2pt);
\node [label={[label distance=0.25cm]5:$\left<a\left<bc\right>_y\right>_x=\left<\left<ab\right>_uc\right>_v$}]  at (e) {}; \filldraw[red] (3.2,1.45) circle (3pt);
\end{tikzpicture}
\captionsetup{font=footnotesize}
\captionof*{figure}{(a) }
\end{minipage}\;\;\;\;
\begin{minipage}[htb]{0.45\textwidth}
\begin{tikzpicture}[dot/.style={circle,inner sep=1pt,fill,label={#1},name=#1},
  extended line/.style={shorten >=-#1,shorten <=-#1},
  extended line/.default=1cm,scale=1.2]

\coordinate (A) at (-2,0);
\coordinate (B) at (2,0);
\coordinate (C) at (0,{sqrt(12)});

\coordinate (ab) at ($(A)!0.5!(B)$); \coordinate (ac) at ($(A)!0.5!(C)$); \coordinate (bc) at ($(C)!0.5!(B)$);
\
\draw[color=blue] (A) -- (B);\draw[] (B) -- (C); \draw[] (C) -- (A);
\
\node [dot=c]  at (ac) {}; \node [dot=a]  at (B) {};
  \node [dot=b]  at (bc) {};
  
  \node [dot,red] at (A) {};  \node [dot,gray] at (C) {};  \node [dot,gray] at (ab) {};

\draw[ very thin,color=blue]
  plot[smooth cycle, tension=1.6]
  coordinates {(ab) (ac) (bc)};
  \node [ label={[label distance=-1.0cm, below]0:$\left<a\left<bc\right>_y\right>_x= \left<\left<ab\right>_uc\right>_v$}] at (A) {};\node [dot] at (B) {}; \node [label={[label distance=0.0cm]0:$\left<ab\right>_u$}] at (C) {}; \node [label={[label distance=0.0cm,below]0:$\left<bc\right>_y$}] at (ab) {};

\end{tikzpicture}
\captionsetup{font=footnotesize}
\captionof*{figure}{(b)  }
\end{minipage}
\caption{ (a) The associativity condition (\ref{associativity22}) is represented as a consistency relation on the Veblen (Menelaus) configuration $(6_2,4_3).$ The points $a,b,c$ (black circles) uniquely define the points (grey circles) $\left<ab\right>_u$ and $\left<bc\right>_y$. Then there are two ways to obtain the point represented by the red circle. The consistency occurs when (\ref{associativity22}) holds.    The  map $(x,y)\mapsto (u,v)$ maps the blue lines  to the black lines of the configuration. (b) The associativity condition (\ref{associativity22})  represented as a consistency relation on the Menelaus configuration where one of its lines is represented by a circle.}\label{Menelaus}

\end{figure}

\tdplotsetmaincoords{70}{200}
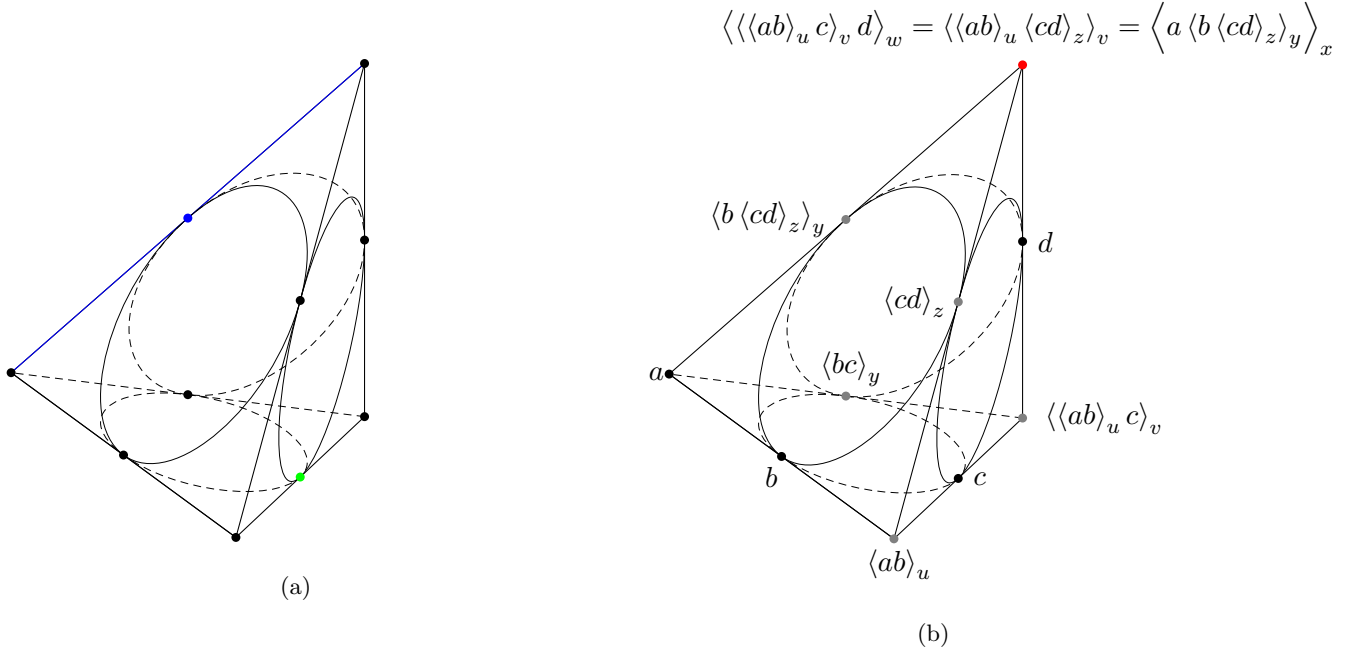
\begin{figure}[htb]
\begin{center}
\begin{minipage}{0.48\textwidth}
\begin{tikzpicture}[tdplot_main_coords,dot/.style={circle,inner sep=1.2pt,fill,label={#1},name=#1},
  extended line/.style={shorten >=-#1,shorten <=-#1},
  extended line/.default=1cm,scale=5]

\coordinate (A) at (1,0,0);
\coordinate (B) at (0,1,0);
\coordinate (C) at (0,0,1);
\coordinate (D) at (0,0,0);

\coordinate (ab) at ($(A)!0.5!(B)$); \coordinate (ac) at ($(A)!0.5!(C)$); \coordinate (bc) at ($(C)!0.5!(B)$);
\coordinate (da) at ($(D)!0.5!(A)$);\coordinate (db) at ($(D)!0.5!(B)$);\coordinate (dc) at ($(D)!0.5!(C)$);
\draw[] (A) -- (B) -- (C) -- cycle;
\draw[densely dashed, very thin] (D)--(A);\draw[] (D)--(B);\draw[] (D)--(C);
\draw[] (A)--(B);
\draw[color=blue] (A)--(C);


\draw[densely dashed, very thin]
  plot[smooth cycle, tension=1.6]
  coordinates {(dc) (da) (ac)};
  
  \draw[densely dashed, very thin]
  plot[smooth cycle, tension=1.6]
  coordinates {(ab) (da) (db)};

  \draw[very thin]
  plot[smooth cycle, tension=1.6]
  coordinates {(dc) (db) (bc)};

\draw[ very thin,shorten >=2pt, shorten <=2pt]
  plot[smooth cycle, tension=1.6]
  coordinates {(ab) (ac) (bc)};

\node [dot] at (dc) {}; \node [dot] at (da) {}; \node [dot,color=blue] at (ac) {};
\node [dot] at (ab) {}; \node [dot,color=green] at (db) {}; \node [dot] at (bc) {};
\node [dot] at (A) {}; \node [dot] at (B) {}; \node [dot] at (C) {}; \node [dot] at (D) {};

\end{tikzpicture}
\captionsetup{font=footnotesize}
\captionof*{figure}{(a)  }
\end{minipage}\hfill
\begin{minipage}{0.48\textwidth}
\begin{tikzpicture}[tdplot_main_coords,dot/.style={circle,inner sep=1.2pt,fill,label={#1},name=#1},
  extended line/.style={shorten >=-#1,shorten <=-#1},
  extended line/.default=1cm,scale=5]

\coordinate (A) at (1,0,0);
\coordinate (B) at (0,1,0);
\coordinate (C) at (0,0,1);
\coordinate (D) at (0,0,0);

\coordinate (ab) at ($(A)!0.5!(B)$); \coordinate (ac) at ($(A)!0.5!(C)$); \coordinate (bc) at ($(C)!0.5!(B)$);
\coordinate (da) at ($(D)!0.5!(A)$);\coordinate (db) at ($(D)!0.5!(B)$);\coordinate (dc) at ($(D)!0.5!(C)$);
\draw[] (A) -- (B) -- (C) -- cycle;
\draw[densely dashed, very thin] (D)--(A);\draw[] (D)--(B);\draw[] (D)--(C);
\draw[] (A)--(B);


\draw[densely dashed, very thin]
  plot[smooth cycle, tension=1.6]
  coordinates {(dc) (da) (ac)};
  
  \draw[densely dashed, very thin]
  plot[smooth cycle, tension=1.6]
  coordinates {(ab) (da) (db)};

  \draw[very thin]
  plot[smooth cycle, tension=1.6]
  coordinates {(dc) (db) (bc)};

\draw[ very thin,shorten >=2pt, shorten <=2pt]
  plot[smooth cycle, tension=1.6]
  coordinates {(ab) (ac) (bc)};

\node [dot, label={[label distance=0.0cm, right]0:$d$} ] at (dc) {}; \node [dot,color=gray, label={[label distance=0.0cm, above]0:$\left<bc\right>_y$}] at (da) {}; \node [dot,color=gray,label={[label distance=-0.2cm, left]0:$\left<b\left<cd\right>_z\right>_{y'}$}] at (ac) {};
\node [dot,label={[label distance=-0.2cm, below]0:$b$}] at (ab) {}; \node [dot,label={[label distance=0.0cm]0:$c$}] at (db) {}; \node [dot,color=gray,label={[label distance=-0.1cm, left]0:$\left<cd\right>_z$}] at (bc) {};
\node [dot,label={[label distance=0.0cm, left]0:$a$}] at (A) {}; \node [dot,color=gray,label={[label distance=0.0cm, below]0:$\left<ab\right>_u$}] at (B) {}; \node [dot,color=red,label={[label distance=0.0cm,above]0:$\left<\left<\left<ab\right>_uc\right>_vd\right>_w=\left<\left<ab\right>_u\left<cd\right>_z\right>_{v'}=
\left<a\left<b\left<cd\right>_z\right>_{y'}\right>_{x}$}] at (C) {}; \node [dot,color=gray,label={[label distance=0.1cm]0:$\left<\left<ab\right>_u c\right>_v$}] at (D) {};

\end{tikzpicture}
\captionsetup{font=footnotesize}
\captionof*{figure}{(b)  }

\end{minipage}

\caption{(a)  The Desargues configuration $(10_3)$  drawn on a tetrahedron. It consists of five Menelaus configurations, one configuration on each face of the tetrahedron and the fifth Menelaus configuration is made by the four circles and the corresponding six points. Three out of five Menelaus configurations share the green point, while the remaining two share the blue line; a manifestation of Pachner $3-2$ move \cite{Pachner:1991}. (b) The pentagon equation represented as a consistency relation on the Desargues configuration. From the points $a,b,c,d$ (black circles) we obtain via the binary operation the points represented by the grey circles. 
 Then the point represented by the red circle is obtained in three different ways. A manifestation of the consistency.}\label{des}

 \end{center}
\end{figure}

The pentagon map $S_I$ was firstly introduced in \cite{Kashaev:1999} inside the context of quantum dilogarithm.
It  also results from the evolution of matrix KP solitons \cite{Dimakis:2018}, or  as a reduction of the so-called {\em normalization map} \cite{Doliwa:2014p,Doliwa:2020}. Moreover, it serves as the top member of the classification list (S-list) of quadrirational pentagon maps \cite{Kassotakis_2024}. 
 In detail, in \cite{Kassotakis_2024} it was proven that any quadrirational \cite{Etingof_1999} pentagon map $S:\cp\times \cp\rightarrow\cp\times\cp$, with $S:(x,y)\mapsto (u,v)$ is $M\ddot{o}b$ equivalent to exactly one
of the following  maps:
\begin{align}\label{S-list}
\begin{aligned}
u=& \frac{x}{x+y-x y }, & v=&x+y-x y, & & &(S_I)\\
u=&x ,& v=&x+y-\delta x y, & & &(S_{II}^\delta)\\
u=&\frac{x}{y}, & v=&y, & & &(S_{III}) \\
u=&x-y, & v=&y, & & &(S_{IV}) 
\end{aligned}
\end{align}
where $\delta=0,1$. Moreover the  associativity
condition (\ref{associativity22}) for the families of binary operations defined by
\begin{align}\label{binary_op}\begin{aligned}
 \left<ab\right>_x^I&:= x\,a+(1-x)\,b, & \left<ab\right>_x^{II}&:= a+(1-\delta
x)\,b,\\
\left<ab\right>_x^{III} &:= x\,a+b, & \left<ab\right>_x^{IV} &:= e^x\,a+b,
\end{aligned}
\end{align}
delivers the respective maps $S_I, S_{II}^{\delta}, S_{III}$ and $S_{IV}.$ 

To introduce the main theorem of this Section, we will need the following definition.
\begin{definition}
A non-empty set $\mathcal{I}$ with a binary operation $\left<\cdot\right>:D\rightarrow \mathcal{I},$ $D\subseteq \mathcal{I}\times \mathcal{I},$ is called partial magma and it is denoted as $(\mathcal{I},\left<\cdot\right>).$ 
\end{definition}

 \begin{theorem}\label{theo01}
Let $\mathcal{I},\mathcal{X}$ be non-empty sets. Let $(\mathcal{I},\left<\cdot\right>_{x})_{x\in \mathcal{X}}$ be a  family of partial  magmas. The map $S:\mathcal{X}\times \mathcal{X}\rightarrow \mathcal{X}\times \mathcal{X},$ $S:(x,y)\mapsto (u,v),$  that satisfies the associativity condition (\ref{associativity22}),  is a pentagon map iff for $a\neq b\neq c\neq d\neq a,$ the equation
\begin{align}\label{refac_1}
  \left<\left<\left<ab\right>_{x'}c\right>_{y'}d\right>_{z'}=\left<\left<\left<ab\right>_{x}c\right>_{y}d\right>_{z},
\end{align}
implies $x'=x,$  $y'=y,$ $z'=z.$
\end{theorem}

\begin{proof}
Let $a,b,c\in \mathcal{I},$ $x,y,u,v\in \mathcal{X}$ and let the map 
\begin{gather*}
S:\mathcal{X}\times \mathcal{X}\ni (x,y)\mapsto (u,v)=\left(f(x,y),g(x,y)\right)\in \mathcal{X}\times \mathcal{X},
\end{gather*}
 satisfy the associativity condition  
\begin{align}\label{assoc_3}
\left<a\left<bc\right>_y\right>_x=& \left<\left<ab\right>_uc\right>_v.
\end{align}
We embed $S$ in $\mathcal{X}\times \mathcal{X}\times \mathcal{X}$, and we have the maps
\begin{align*}
S_{12}:(x,y,z)\mapsto & (\hat x,\tilde y,z)=\left(f(x,y),g(x,y),z\right),\\
S_{13}:(x,y,z)\mapsto & (\bar x,y,\tilde z)=\left(f(x,z),y,g(x,z)\right),\\
S_{23}:(x,y,z)\mapsto & (x,\bar y,\hat z)=\left(x,f(y,z),g(y,z)\right).
\end{align*}
Considering now the following points of $\mathcal{I}:$ $\left<a\left<b\left<cd\right>_z\right>_y\right>_x,$  $\left<a\left< \left<bc\right>_y d\right>_z\right>_x,$ 
$\left< \left<a\left<bc\right>_y\right>_x d\right>_z.$
By using (\ref{assoc_3}), clearly the maps $S_{23}$, $S_{13}$, $S_{12},$ respectively satisfy
  \begin{align}
  \begin{aligned}\label{as1_1}
  \left<a\left<b\left<cd\right>_z\right>_y\right>_x=&\left<a\left< \left<bc\right>_{\bar y} d\right>_{\hat z}\right>_x,\\ 
   \left<a\left< \left<bc\right>_y d\right>_z\right>_x=&\left< \left<a\left<bc\right>_y\right>_{\bar x} d\right>_{\tilde z},\\ 
   \left< \left<a\left<bc\right>_y\right>_x d\right>_z=&\left<\left<\left<ab\right>_{\hat x}c\right>_{\tilde y}d\right>_z.
   \end{aligned}
  \end{align}
 So $S_{13}S_{23}:(x,y,z)\mapsto  (\bar x,\bar y,\tilde {\hat z}),$   satisfies
 \begin{align}  \label{lhsp2}
 \left<a\left<b\left<cd\right>_z\right>_y\right>_x=&\left< \left<a\left<bc\right>_{\bar y}\right>_{\bar x} d\right>_{\tilde{ \hat z}},
 \end{align}
 and 
 \begin{align}\label{plhs}
 S_{12}S_{13}S_{23}:(x,y,z)\mapsto  (\hat{\bar x},\tilde{\bar y},\tilde {\hat z}),
 \end{align}
  satisfies
 \begin{align}\label{lhs_pp}
 \left<a\left<b\left<cd\right>_z\right>_y\right>_x=&\left<\left<\left<ab\right>_{\hat{\bar x}}c\right>_{\tilde{\bar y}}d\right>_{\tilde {\hat z}}.
 \end{align}
Now if we consider the points of $\mathcal{I}:$ $\left<a\left<b\left<cd\right>_z\right>_y\right>_x,$ $\left<\left<ab\right>_x\left<cd\right>_z\right>_y$ 
From (\ref{assoc_3}),  the maps $S_{12}$, $S_{23}$, respectively satisfy
  \begin{align}
  \begin{aligned} \label{as1_2}
 \left<a\left<b\left<cd\right>_z\right>_y\right>_x=&\left<\left<ab\right>_{\hat x}\left<cd\right>_z\right>_{\tilde y},\\
 \left<\left<ab\right>_x\left<cd\right>_z\right>_y=&\left<\left<\left<ab\right>_{x}c\right>_{\bar y}d\right>_{\hat z}.
 \end{aligned}
  \end{align}
 So 
 \begin{align}\label{prhs}
 S_{23}S_{12}:(x,y,z)\mapsto (\hat x,\bar{\tilde y},\hat z)
 \end{align}
  satisfies
 \begin{align} \label{rhs_pp}
 \left<a\left<b\left<cd\right>_z\right>_y\right>_x=&\left<\left<\left<ab\right>_{\hat x}c\right>_{\bar{\tilde y}}d\right>_{\hat z}.
 \end{align}
 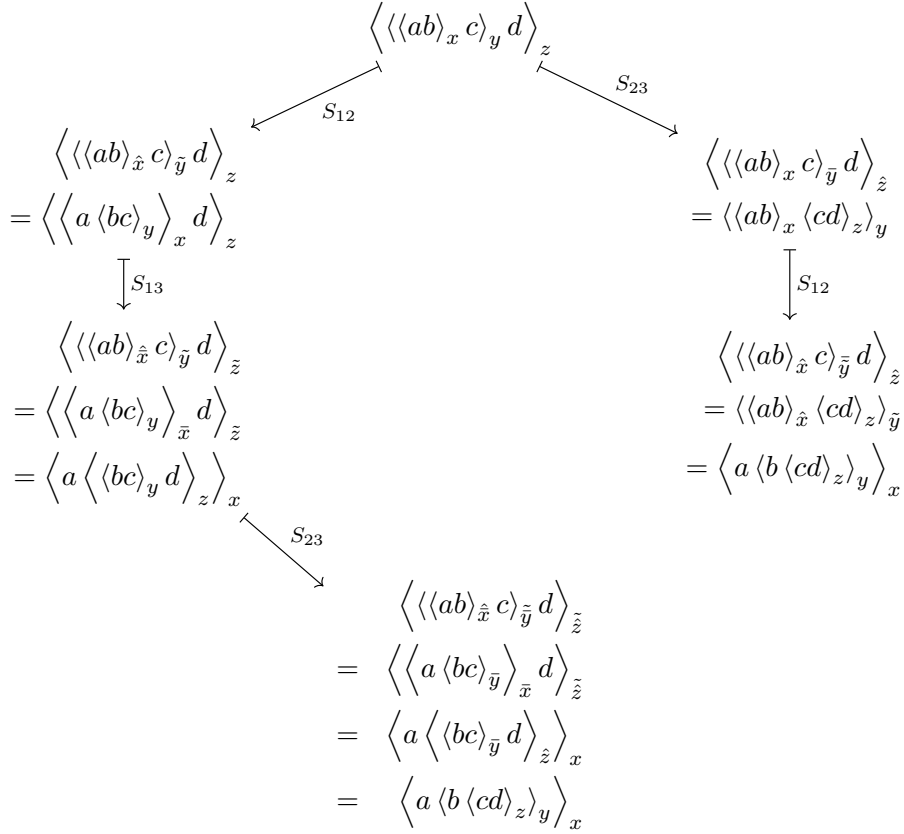
\begin{figure}[htb]
 \begin{tikzcd}
  {} & \left<\left<\left<ab\right>_{ x}c\right>_{y}d\right>_{z} \arrow[dl, |->, "S_{12}"] \arrow[dr, |->, "S_{23}"] & {}\\
  {\begin{aligned}
  \left<\left<\left<ab\right>_{\hat x}c\right>_{\tilde y}d\right>_{z}\\
  =\left< \left<a\left<bc\right>_y\right>_x d\right>_z 
  \end{aligned} \arrow[d, |->, "S_{13}"]}& {}&  {\begin{aligned}\left<\left<\left<ab\right>_{ x}c\right>_{\bar y}d\right>_{\hat z} \\ =\left<\left<ab\right>_x\left<cd\right>_z\right>_y\end{aligned} \arrow[d, |->, "S_{12}"]}  \\
  {\begin{aligned}
  \left<\left<\left<ab\right>_{\hat{\bar x}}c\right>_{\tilde y}d\right>_{\tilde z}\\
  =\left< \left<a\left<bc\right>_y\right>_{\bar x} d\right>_{\tilde z}\\
  {} =\left<a\left< \left<bc\right>_y d\right>_z\right>_x
  \end{aligned}} \arrow[dr, |->, "S_{23}"]  &{}  {}&{\begin{aligned}\left<\left<\left<ab\right>_{\hat x}c\right>_{\bar{\tilde y}}d\right>_{\hat z}\\
  =\left<\left<ab\right>_{\hat x}\left<cd\right>_z\right>_{\tilde y}\\
  {}=\left<a\left<b\left<cd\right>_z\right>_y\right>_x\end{aligned}}\\
  {} &{\begin{aligned}
  &\left<\left<\left<ab\right>_{\hat{\bar x}}c\right>_{\tilde{\bar y}}d\right>_{\tilde{\hat z}}\\
  =&\left< \left<a\left<bc\right>_{\bar y}\right>_{\bar x} d\right>_{\tilde{\hat z}}\\
  =&\left<a\left< \left<bc\right>_{\bar y} d\right>_{\hat z}\right>_x\\
  =&\left<a\left<b\left<cd\right>_z\right>_y\right>_x
  \end{aligned}}  & {}
 \end{tikzcd}
 \caption{The chain of maps $S_{12}S_{13}S_{23}$ and $S_{23}S_{12}$ acting on the point $\left<\left<\left<ab\right>_{ x}c\right>_{y}d\right>_{z}\in \mathcal{I}.$ The map $S$ satisfies the associativity condition (\ref{associativity22}). } \label{fig_lr}
 \end{figure}

 \begin{figure}[htb]
 \begin{tikzcd}
  {} & \left<\left<\left<ab\right>_{ x}c\right>_{y}d\right>_{z} \arrow[dl, |->, "S_{12}"] \arrow[dr, |->, "S_{23}"] & {}\\
  {\left< \left<a\left<bc\right>_y\right>_x d\right>_z 
   \arrow[d, |->, "S_{13}"]}& {}&  { \left<\left<ab\right>_x\left<cd\right>_z\right>_y  \arrow[d, |->, "S_{12}"]}  \\
  {\left<a\left< \left<bc\right>_y d\right>_z\right>_x
  } \arrow[rr, |->, "S_{23}"]  &{}  {}&{   \left<a\left<b\left<cd\right>_z\right>_y\right>_x}\\
  {} &{}  & {}
 \end{tikzcd}
 \caption{The diagram of Figure \ref{fig_lr} when $S$ is a pentagon map. The pentagonal identity.} \label{fig_pe}
 \end{figure}

 On the one hand, if $S$ is a pentagon map so $S_{12}S_{13}S_{23}=S_{23}S_{12},$ from (\ref{plhs}),(\ref{prhs})   we have 
 \begin{align*}
   (\hat{\bar x},\tilde{\bar y},\tilde {\hat z})= & (\hat x,\bar{\tilde y},\hat z),
 \end{align*}
 and equation (\ref{lhs_pp}) coincides with (\ref{rhs_pp}) (see also Figure \ref{fig_pe}). On the other hand (see Figure \ref{fig_lr}), in order $S$ to be a pentagon map, from (\ref{lhs_pp}), (\ref{rhs_pp}) there should be that the equation
 \begin{align}\label{refa1}
 \left<\left<\left<ab\right>_{\hat{ \bar x}}c\right>_{\tilde{ \bar y}}d\right>_{\tilde{ \hat z}}=\left<\left<\left<ab\right>_{\hat x}c\right>_{\bar{\tilde y}}d\right>_{\hat z},
  \end{align}
 should imply 
\begin{align*}
  (\hat{\bar x},\tilde{\bar y},\tilde {\hat z})= & (\hat x,\bar{\tilde y},\hat z). 
\end{align*}
So the condition
\begin{align} \label{rfp}
\left<\left<\left<ab\right>_{x'}c\right>_{y'}d\right>_{z'}=&\left<\left<\left<ab\right>_{x}c\right>_{y}d\right>_{z}&&\implies& x'=&x,& y'=&y,&z'=&z,
\end{align}
is equivalent to the pentagon equation
 and that completes the proof.

\end{proof}
Theorem \ref{theo01}, provides the necessary and sufficient conditions for a map that satisfies (\ref{associativity22}) to be a pentagon map. In this respect, (\ref{rfp}) serves as a translation of the so-called {\em three-factorization property} \cite{Kouloukas2009} (associated with Yang-Baxter maps) and the {\em six-factorization property} \cite{Rizos:2022} (associated with tetrahedron maps) to the pentagon setting.

We have the following Corollary. 


\begin{corollary}\label{cor1}
Let $\mathcal{I},\mathcal{X}$ be non-empty sets. Let $(\mathcal{I},\left<\cdot\right>)_{x\in \mathcal{X}}$ be a family of partial magmas satisfying the bi-injectivity property:
\begin{align*}
\left<a'b\right>_{x'}=\left<ab\right>_x \implies &\; x'=x,\; a'=a.
\end{align*}
 Then a map $S:\mathcal{X}\times \mathcal{X}\rightarrow \mathcal{X}\times \mathcal{X},$ $S:(x,y)\mapsto (u,v),$  that satisfies the associativity condition (\ref{associativity22}),  is a pentagon map.
\end{corollary}

\begin{proof}
According to Theorem \ref{theo01}, we have to show that 
\begin{align}\label{pcor}
\left<\left<\left<ab\right>_{x'}c\right>_{y'}d\right>_{z'}=\left<\left<\left<ab\right>_{x}c\right>_{y}d\right>_{z},
\end{align}
implies $x'=x,$ $y'=y,$ $z'=z.$ The family of magmas $(\mathcal{I},\left<\cdot\right>)_{x\in \mathcal{X}}$, admits the bi-injectivity property, so (\ref{pcor}) implies $z'=z$ and 
\begin{align*}
 \left<\left<ab\right>_{x'}c\right>_{y'}=\left<\left<ab\right>_{x}c\right>_{y},
\end{align*}
that in turn implies $y'=y$ and 
\begin{align*}
 \left<ab\right>_{x'}=\left<ab\right>_{x},
\end{align*}
that finally implies $x'=x.$ So the requirements of Theorem \ref{theo01} are satisfied and $S$ is a pentagon map.
\end{proof}
As an example we study the family of binary operations
\begin{align*}
  \left<ab\right>_{x}:= & f(x)a+x b,
\end{align*}
and we provide the necessary and sufficient conditions that $f$ should satisfy so that a map that is equivalent to the associativity condition (\ref{associativity22}) to be a pentagon map.
In detail, although we could use Corollary \ref{cor1}, here we use Theorem \ref{theo01} to prove the first part of the following Proposition. 
\begin{prop} \label{prop01}
A. Let $S:\mathcal{X}\times \mathcal{X}\rightarrow \mathcal{X}\times \mathcal{X},$ be the map 
\begin{align*}
  S: (x,y)\mapsto & (u,v)=\left(x\frac{f(y)}{f(xy)},xy\right),
  \end{align*}
where $f:\mathcal{X}\rightarrow \mathcal{X}$ satisfies 
\begin{align}\label{fcon} 
f(u)f(v)=&f(x).
\end{align}
Then $S$ is a pentagon map.

B. The function $f:\mathcal{X}\rightarrow \mathcal{X},$ defined by the  curve
\begin{align}\label{fermat} 
x^\alpha+(f(x))^\alpha=&1,&\alpha>0\in\mathbb{R},
\end{align}
where $x\in\mathbb{CP}^1$ when $\alpha=1,$ $x\in (0,1)\subset \mathbb{R}$ when $\alpha\neq 1,$ satisfies (\ref{fcon}). The associated pentagon map $S,$ preserves the Poisson structure 
\begin{align*}
\Omega=&m(x,y)\partial x\wedge \partial y,& m(x,y):=&y(f(y))^{\alpha-1}(f(x))^\alpha,
\end{align*} 
 and admits the function $I(x,y):=x f(y)$ as an invariant, hence it is a Liouville integrable map.
\end{prop}

\begin{proof}
A. We consider the family of binary operations 
\begin{align}\label{bin1}
\left<ab\right>_x:= f(x) a+x b.
\end{align} 
The associativity condition (\ref{associativity22}), in terms of the binary operation (\ref{bin1}) reads
\begin{align}\label{asc}
  f(u)f(v)\;a+u f(v)\;b+v\; c= & f(x)\;a+xf(y)\;b+xy\; c.
\end{align}
From the coefficients of $c$ and of  $b$ we obtain  
\begin{align*}
v=&xy,& u=&x\frac{f(y)}{f(v)}=x\frac{f(y)}{f(xy)},
\end{align*}
which are  exactly the defining relations of the map $S$ of the Proposition, while the coefficient of $a$ reads $f(u)f(v)=f(x),$ that is exactly the  condition (\ref{fcon}). So the map $S$ is equivalent to the associativity condition (\ref{associativity22}).

From Theorem \ref{theo01}, a map $S$ that satisfies (\ref{associativity22}) for the family of binary operations (\ref{bin1}), is a pentagon  map if the equation  (\ref{refac_1}) implies the identity solution. In terms of the binary operation (\ref{bin1}), equation (\ref{refac_1}) reads
\begin{multline*}
f(z')f(y')f(x')\;a+x'f(z')f(y')\;b+f(z')y'\;c+z'\;d\\=f(z)f(y)f(x)\;a+xf(z)f(y)\;b+f(z)y\;c+z\;d,
\end{multline*}
that results the following system of equations
\begin{align*}
z'=&z, & f(z')y'&=f(z)y,& x'f(z')f(y')=&xf(z)f(y),& f(z')f(y')f(x')=&f(z)f(y)f(x).
\end{align*}  
The equations above, clearly imply the identity solution $x'=x,$ $y'=y$ and $z'=z.$

 So provided that the  map $S$ satisfies (\ref{fcon}), it is equivalent to (\ref{associativity22}), and since (\ref{refac_1}) implies the identity solution, from Theorem \ref{theo01} we conclude that $S$ is a pentagon map.

B. Now we prove that $f$ defined by (\ref{fermat}) satisfies (\ref{fcon}). Using the definitions of $u$, and  $v$ from the map, we have 
\begin{align*}
f(u)f(v)=&f\left(x\frac{f(y)}{f(v)}\right)f(v)=\left(1-x^\alpha\frac{(f(y))^\alpha}{(f(v))^\alpha}\right)^{\frac{1}{\alpha}}f(v)=
\left((f(v))^\alpha-x^\alpha(f(y))^\alpha\right)^{\frac{1}{\alpha}},
\end{align*}
that holds due to the assumptions of the Proposition on the domain of $f.$ Now since
Since  $(v,f(v)),$  $(y,f(y)),$ are points on the curve we have 
\begin{align*}
\left((f(v))^\alpha-x^\alpha(f(y))^\alpha\right)^{\frac{1}{\alpha}}=&\left(1-v^\alpha-x^\alpha(1-y^\alpha)\right)^{\frac{1}{\alpha}}=f(x).
\end{align*}
So $f(u)f(v)=f(x),$ and indeed this choice of $f$ satisfies the  condition (\ref{fcon}).

 We turn now to the proof that when $f$ is defined from (\ref{fermat}), $S$ is a Liouville integrable map.  The function $I(x,y):=x f(y),$ is an invariant function of $S$ since it holds 
    \begin{align*}
    I(u,v)=u f(v)=xf(y)=I(x,y).
    \end{align*}
    In order to prove that $S$ preserves the Poisson structure $\Omega$ we have to show that $m(u,v) \partial u\wedge \partial v=m(x,y) \partial x\wedge \partial y,$ or equivalently to show that 
    \begin{align}\label{mes}
    \frac{m(u,v)}{m(x,y)}=\left|\frac{\partial (u,v)}{\partial (x,y)}\right|,
    \end{align} 
    where $\left|\frac{\partial (u,v)}{\partial (x,y)}\right|$ denotes the Jacobian determinant of the map.
    There is 
    \begin{align*}
    \left|\frac{\partial (u,v)}{\partial (x,y)}\right|=&\begin{vmatrix}
                                                              {\displaystyle  \frac{\partial u}{\partial x}} & {\displaystyle\frac{\partial u}{\partial y}} \\
                                                               {\displaystyle \frac{\partial v}{\partial x}} & {\displaystyle\frac{\partial v}{\partial y} } 
                                                               \end{vmatrix}=\frac{x}{f(xy)}\left(f(y)-yf'(y)\right),
    \end{align*}
    where $f'(y)$ denotes the derivative of $f(y)$ with respect to $y$. From the equation of the curve $y^\alpha+(f(y))^\alpha=1,$ by differentiation we obtain $y^{\alpha-1}+(f(y))^{\alpha-1}f'(y)=0,$ and the Jacobian determinant finally reads
\begin{align} \label{jac01}
    \left|\frac{\partial (u,v)}{\partial (x,y)}\right|=&\frac{x}{f(xy)(f(y))^{\alpha-1}}.
    \end{align}    
  On the other hand we have
  \begin{multline*}
  m(u,v)=v(f(v))^{\alpha-1}(f(u))^\alpha=\frac{v(f(v))^\alpha}{f(v)}(1-u^\alpha)=
  \frac{xy(f(xy))^\alpha}{f(xy)}\left(1-x^\alpha\frac{(f(y))^\alpha}{(f(xy))^\alpha}\right)\\
  =\frac{xy}{f(xy)}(f(x))^\alpha,
  \end{multline*}  
and
  \begin{align*}
  \frac{m(u,v)}{m(x,y)}=&\frac{\frac{xy}{f(xy)}(f(x))^\alpha}{y(f(y))^{\alpha-1}(f(x))^\alpha}=\frac{x}{f(xy)(f(y))^{\alpha-1}},
  \end{align*}
  that coincides with (\ref{jac01}), that assures that the Poisson structure $\Omega$ is preserved under the action of the map $S$. To recapitulate, mapping $S,$ preserves $\Omega$ and admits $I(x,y)$ as an invariant function, so it is a Liouville integrable map according to the definition of  Liouville integrability of maps in \cite{ves2}.

\end{proof}

For the pentagon map $S$
of Proposition \ref{prop01} the following remarks are in order.
\begin{itemize}
\item The map $S,$ in its equivalent form $\tau S\tau$ (where $\tau:(x,y)\mapsto (y,x)$)  that  satisfies the reverse-pentagon equation,    first appeared in \cite{Kashaev:1998}. 
\item The map $S$  represents a family of Liouville integrable pentagon maps parameterized by $\alpha>0\in\mathbb{R}.$ This family of maps does not belong to the QRT family \cite{Quispel1988} of Liouville integrable maps. It is a novel family of Liouville integrable HKY (non-QRT)  type \cite{Hirota2001} maps, as it appears to be absent from the relevant literature  \cite{Tanaka2009,Kassotakis2010,Roberts_2015,Joshi2019}.
\item For $\alpha=1,$ $S$ is  a rational pentagon map which is (M\"ob) equivalent to $S_I.$
For $\alpha=2,$ we have a trigonometric pentagon map that explicitly reads 
 \begin{align*}
 S: (x,y)\mapsto & \left(x\frac{\sqrt{ 1-y^2}}{\sqrt{ 1-x^2y^2}},xy\right),
 \end{align*}
 or equivalently
 \begin{align*}
 S: ({\sin \theta},{\sin \phi})\mapsto & \left({\sin \theta}\frac{{\cos \phi}}{\cos (\arcsin(\sin \theta \sin \phi))},\sin \theta \sin \phi\right).
 \end{align*}
 For $\alpha=3$ an elliptic map while for $\alpha>3\in\mathbb{N}$ we obtain  elliptic maps of higher genus.
\item The second component of $S,$ serves as an addition theorem on the corresponding  curve. 
\item The inverse of $S$ reads 
\begin{align*}
  S^{-1}:(x,y)\mapsto & \left(f\left(f(x)f(y)\right),\frac{y}{f\left(f(x)f(y)\right)}\right).
\end{align*}
Moreover $S$ together with $S^{-1}$ satisfy the ten-term relation \cite{Kashaev:1998}, that results that the map
\begin{align*}
  T:=S^{-1}_{13}\tau_{23} S_{13}:(x,y,z)\mapsto & \left(f\left(f\left(x\frac{f(z)}{f(xz)}\right)f(y)\right),xz,\frac{y}{f\left(f\left(x\frac{f(z)}{f(xz)}\right)f(y)\right)}\right),
\end{align*} 
 satisfies the tetrahedron equation
\begin{align*}
T_{123}T_{145}T_{246}T_{356}=&T_{356}T_{246}T_{145}T_{123}.
\end{align*}
\item From  $S \tau $ and $\tau S,$ where $\tau: (x,y)\mapsto (y,x),$ we can build the map 
\begin{align}\label{hexa}
H:(x,y)\mapsto & \left(y\frac{f(x)}{f(xy)},xy,x\frac{f(y)}{f(xy)}\right),& z^\alpha+(f(z))^\alpha=&1, &\alpha>0,
\end{align} 
 that it can be shown that satisfies the hexagon equation 
\begin{align*}
H_{12}H_{23}\tau_{34}H_{12}=&\tau_{34}H_{45}H_{23}\tau_{12}H_{23},
\end{align*} 
so it is a hexagon map.
The form of the hexagon equation above is presented in \cite{Dimakis:2015}, and the map $H$ (for $\alpha=1$) was first presented in \cite{Kashaev2015}, see also \cite{Dimakis:2018} where the construction above (for $\alpha=1$) was given.
\end{itemize}

\section{$n-$ary operations and multicomponent pentagon maps}\label{sec3_0}

If we consider  families of partial $n-$ary magmas $(\mathcal{I},\left<\cdot\right>_{\bf{ x}})_{x\in \mathcal{X}^{n-1}},$ where $\left<\cdot\right>_{\bf{ x}}:\mathcal{I}^n\rightarrow \mathcal{I},$   $n>2,$ a family of $n-$ary operations, we have the following Theorem that extends Theorem \ref{theo01}.

 \begin{theorem}\label{theo11}
Let $\mathcal{I},\mathcal{X}$ be non-empty sets. Let $(\mathcal{I},\left<\cdot\right>_{\bf{ x}})_{x\in \mathcal{X}^{n-1}},$ $n\geq 2$ be a  family of  $n$-ary partial magmas. The map $S:\mathcal{X}^{n-1}\times \mathcal{X}^{n-1}\rightarrow \mathcal{X}^{n-1}\times \mathcal{X}^{n-1},$ $S: ({\bf x}, {\bf y})\mapsto ({\bf u}, {\bf v})$ where ${\bf x}:=(x_1\ldots, x_{n-1}),$ ${\bf y}:=(y_1,\ldots, y_{n-1}),$ ${\bf u}:=(u_1,\ldots, u_{n-1}),$ ${\bf v}:=(v_1,\ldots, v_{n-1}),$ that satisfies the  condition 
\begin{align}\label{nary_ass}
\left<\left<a_1\ldots a_n\right>_{\bf u} a_{n+1}\ldots a_{2n-1}\right>_{\bf v}=&\left<a_1\ldots a_{n-1}\left<a_n\ldots a_{2n-1}\right>_{\bf y}\right>_{\bf x},
\end{align}

  is a pentagon map iff for $a_i\neq a_j,$  the equation
\begin{multline} \label{refac_11}
  \left<\left<\left<a_1\ldots a_n\right>_{{\bf x}'}a_{n+1}\ldots a_{2n-1}\right>_{{\bf y}'}a_{2n}\ldots a_{3n-1}\right>_{{\bf z}'}\\=
  \left<\left<\left<a_1\ldots a_n\right>_{{\bf x}}a_{n+1}\ldots a_{2n-1}\right>_{{\bf y}}a_{2n}\ldots a_{3n-1}\right>_{{\bf z}},
\end{multline}
implies ${\bf x}'={\bf x}, {\bf y}'={\bf y}, {\bf z}'={\bf z}.$
\end{theorem}

\begin{proof}
First note that for $n=2$, (\ref{nary_ass}) and (\ref{refac_11}) coincide respectively with (\ref{associativity22}) and (\ref{refac_1}), so Theorem \ref{theo01} is included.

For $n>2$, the map on the left hand side (LHS) of the pentagon equation, that is:
\begin{align}\label{Plhs}
 S_{12}S_{13}S_{23}:(\mathbf{ x},\mathbf{ y},\mathbf{ z})\mapsto  (\hat{\bar {\mathbf{ x}}},\tilde{\bar {\mathbf{ y}}},\tilde {\hat {\mathbf{ z}}}),
 \end{align}
  satisfies
  \begin{multline*}
   (LHS): \left<\left<\left<a_1\ldots a_n\right>_{\hat{\bar {\mathbf{ x}}}}a_{n+1}\ldots a_{2n-1}\right>_{\tilde{\bar {\mathbf{ y}}}}a_{2n}\ldots a_{3n-1}\right>_{\tilde{ \hat {\mathbf{ z}}}} \\
    =\left<\left<a_1\ldots a_{n-1}\left<a_n\ldots a_{2n-1}\right>_{\bar{\mathbf{ y}}}\right>_{\bar {\mathbf{ x}}}a_{2n}\ldots a_{3n-1}\right>_{\tilde{\hat { \mathbf{z}}}} \\
     =\left<a_1\ldots a_{n-1}\left<\left<a_n\ldots a_{2n-1}\right>_{\bar{\mathbf{ y}}}a_{2n}\ldots a_{3n-1} \right>_{\hat { \mathbf{z}}} \right>_{ \mathbf{x}}\\
     =\left<a_1\ldots a_{n-1}\left<a_n\ldots a_{2n-2}\left<a_{2n-1}\ldots a_{3n-1}\right>_{\mathbf{z}} \right>_{\mathbf{y}} \right>_{ \mathbf{x}}. 
  \end{multline*}
  While the map on the right hand side (RHS) of the pentagon equation, that is:
\begin{align}\label{Plhs}
 S_{23}S_{12}:(\mathbf{ x},\mathbf{ y},\mathbf{ z})\mapsto  (\hat{ \mathbf{x}},\bar{\tilde{\mathbf y}},\hat{ \mathbf{z}}),
 \end{align}
 satisfies
\begin{multline*}
   (RHS): \left<\left<\left<a_1\ldots a_n\right>_{\hat{ \mathbf{x}}}a_{n+1}\ldots a_{2n-1}\right>_{\bar{\tilde{\mathbf y}}}a_{2n}\ldots a_{3n-1}\right>_{\hat{ \mathbf{z}}} \\
  =\left< \left<a_1\ldots a_n\right>_{\hat{ \mathbf{x}}}\ldots a_{2n-2}\left<a_{2n-1}\ldots a_{3n-1} \right>_{\mathbf{ z}}\right>_{\tilde{\mathbf{y}}}\\
     =\left<a_1\ldots a_{n-1}\left<a_n\ldots a_{2n-2}\left<a_{2n-1}\ldots a_{3n-1}\right>_{\mathbf{z}} \right>_{\mathbf{y}} \right>_{ \mathbf{x}}. 
  \end{multline*}
  So if $S$ is a pentagon map the (LHS) coincides with the (RHS). On the other side, for $S$ to be a pentagon map (LHS) should coincide with the (RHS) that results that the equation
  \begin{align*}
  \left<\left<\left<a_1\ldots a_n\right>_{\hat{\bar {\mathbf{ x}}}}a_{n+1}\ldots a_{2n-1}\right>_{\tilde{\bar {\mathbf{ y}}}}a_{2n}\ldots a_{3n-1}\right>_{\tilde{ \hat {\mathbf{ z}}}}=
  \left<\left<\left<a_1\ldots a_n\right>_{\hat{ \mathbf{x}}}a_{n+1}\ldots a_{2n-1}\right>_{\bar{\tilde{\mathbf y}}}a_{2n}\ldots a_{3n-1}\right>_{\hat{ \mathbf{z}}}
  \end{align*}
  should imply 
  \begin{align*}
  (\hat{\bar {\mathbf{ x}}},\tilde{\bar {\mathbf{ y}}},\tilde {\hat {\mathbf{ z}}})=&(\hat{ \mathbf{x}},\bar{\tilde{\mathbf y}},\hat{ \mathbf{z}}),
  \end{align*}
  and that completes the proof.
\end{proof}
\begin{corollary} \label{cor2}
Let $\mathcal{I},\mathcal{X}$ be non-empty sets. Let $(\mathcal{I},\left<\cdot\right>_{\bf{ x}})_{x\in \mathcal{X}^{n-1}},$ $n\geq 2$ be a  family of  $n$-ary partial magmas satisfying the injectivity property:
\begin{align}\label{bij2}
\left<a_1'a_2\ldots a_{n-1}\right>_{\mathbf{ x}'}=\left<a_1a_2\ldots a_{n-1}\right>_{\mathbf{ x}} \implies & \mathbf{ x}'=\mathbf{ x},\; a_1'=a_1.
\end{align}
 Then a map $S:\mathcal{X}^{n-1}\times \mathcal{X}^{n-1}\rightarrow \mathcal{X}^{n-1}\times \mathcal{X}^{n-1},$ $S:({\mathbf x},{\mathbf y})\mapsto ({\mathbf u},{\mathbf v}),$  that satisfies the condition (\ref{nary_ass}),  is a pentagon map.
\end{corollary}

\begin{proof}
The proof is similar to the proof of Corollary \ref{cor1}. According to Theorem \ref{theo11}, we have to show that 
\begin{multline} \label{pcor2}
  \left<\left<\left<a_1\ldots a_n\right>_{{\bf x}'}a_{n+1}\ldots a_{2n-1}\right>_{{\bf y}'}a_{2n}\ldots a_{3n-1}\right>_{{\bf z}'}\\=
  \left<\left<\left<a_1\ldots a_n\right>_{{\bf x}}a_{n+1}\ldots a_{2n-1}\right>_{{\bf y}}a_{2n}\ldots a_{3n-1}\right>_{{\bf z}},
\end{multline}
implies ${\mathbf x}'={\mathbf x},$ ${\mathbf y}'={\mathbf y},$ ${\mathbf z}'={\mathbf z}.$ The family of partial magmas $(\mathcal{I},<\;>_{\bf{ x}})_{x\in \mathcal{X}^{n-1}},$ admits the injectivity property, so (\ref{pcor2}) implies that ${\mathbf z}'={\mathbf z}$ and 
\begin{align*}
 \left<\left<a_1\ldots a_n\right>_{{\bf x}'}a_{n+1}\ldots a_{2n-1}\right>_{{\bf y}'}=&\left<\left<a_1\ldots a_n\right>_{{\bf x}}a_{n+1}\ldots a_{2n-1}\right>_{{\bf y}},
\end{align*}
that in turn implies ${\mathbf y}'={\mathbf y}$ and 
\begin{align*}
 \left<a_1\ldots a_n\right>_{{\bf x}'}=&\left<a_1\ldots a_n\right>_{{\bf x}},
\end{align*}
that finally implies ${\mathbf x}'={\mathbf x}.$ So the requirements of Theorem \ref{theo11} are satisfied and $S$ is a pentagon map.

\end{proof}
We mention here that a possible generalization of Theorem \ref{theo01} and of Theorem \ref{theo11}, is to consider families of $m-$valued magmas and families of $m-$valued $n-$ary magmas respectively. Note that $m-$valued operations on sets were first introduced in \cite{BuchstaberNovikov1971}, see also the recent work on $m-$valued quandles \cite{Talalaev2024}.

\subsection{Ternary operations and  pentagon maps}

In this Section, we apply Corollary \ref{cor2} to the ternary case for a specific but quite general family of ternary operations. In detail, we study the family of ternary operations
\begin{align*}
  \left<abc\right>_{x,X}:= & f(x,X)a+g(x)b+Xc. 
\end{align*}
In the following Proposition, by using  Corollary \ref{cor2} we provide the necessary  conditions that $f$ and $g$ should satisfy so that a map that satisfies (\ref{nary_ass}) with $n=3,$   to be a pentagon map. Moreover we present two  choices of families of $f,g$ which result  two families of two-component pentagon maps.

\begin{prop}\label{prop3}
The map $S: \mathcal{X}^2\times \mathcal{X}^2\rightarrow \mathcal{X}^2\times \mathcal{X}^2$ with 
\begin{align*}
S:(x,X;y,Y)\mapsto& (u,U;v,V)=\left(g^{-1}\left(\frac{g(x)}{f(v,V)}\right),X\frac{f(y,Y)}{f(v,V)};g^{-1}\left(Xg(y)\right),X Y\right),
\end{align*}
where $g:\mathcal{X}\rightarrow \mathcal{X},$  a bijection and $f:\mathcal{X}\times \mathcal{X}\rightarrow \mathcal{X}\times \mathcal{X},$ satisfies 
\begin{align}\label{con_te1}
f(u,U)f(v,V)=f(x,X),
\end{align}
is a pentagon map.
\end{prop}

\begin{proof}
We consider the family of ternary operations 
\begin{align}\label{tern1}
\left<abc\right>_{x,X}:= & f(x,X)a+g(x)b+Xc.
\end{align} 
The  condition (\ref{nary_ass}) reads
\begin{align} \label{assoc_ter}
  \left<\left<abc\right>_{u,U}de\right>_{v,V} =& \left<ab\left<cde\right>_{y,Y}\right>_{x,X},
\end{align}
or
\begin{align*}
  f(v,V)\left(f(u,U) a+g(u) b+U c\right)+g(v) d+V e=&f(x,X) a+g(x) b+X\left(f(y,Y) c+g(y) d+Y e\right).
\end{align*}
From the coefficients of  $e,d,c,$ and $b,$ of the equation above we obtain
\begin{align*}
  V & =X Y, & g(v)=&Xg(y),& U=&X\frac{f(y,Y)}{f(v,V)},& g(u)=&\frac{g(x)}{f(v,V)},
\end{align*}
that provided that $g$ is a bijection, we obtain exactly the defining relations for the map $S$ of the Proposition. Furthermore, the coefficient of $a$ reads $f(u,U)f(v,V)=f(x,X),$ that is exactly the  condition (\ref{con_te1}).

So  provided that there is an  $f$ such that (\ref{con_te1}) holds and $g$ is a bijection, mapping $S$ is equivalent to (\ref{assoc_ter}). From Corollary \ref{cor2},  mapping $S$ is a pentagon map if the family of ternary operation (\ref{tern1}) respects the injectivity property (\ref{bij2}). In terms of the ternary operation (\ref{tern1}), the first equation of (\ref{bij2}) reads
\begin{align} \label{bije2}
  \left<a'bc\right>_{x',X'} & =\left<abc\right>_{x,X},
\end{align}
or 
\begin{align}\label{bije22}
  f(x',X')a'+g(x')b+X'c & =f(x,X)a'+g(x)b+X c.
\end{align}
From the coefficient of $c,$ we have $X'=X.$ From the coefficient of $b$ we have $g(x')=g(x),$ that since $g$ is a bijection it implies $x'=x.$  So we have $X'=X,$ $x'=x$ and  (\ref{bije22}) becomes $ f(x,X)a'=f(x,X)a,$ that implies $a'=a$ provided that $f$ is not the trivial zero function.  So indeed the ternary operation (\ref{tern1}) respects the required injectivity property and the map $S$ satisfies (\ref{con_te1}), hence it is a pentagon map.  
\end{proof}

The family of maps $S$ of the previous Proposition is not empty. There exist bijections $g$ and functions $f$ such that the functional equation (\ref{con_te1}) is satisfied. For example, if we consider $g$ to be defined by the  curve 
\begin{align}\label{g-f}
  x^\alpha+(g(x))^\alpha & =1,& \alpha>&0\in\mathbb{R}, & \begin{aligned}
                       x\in  \mathbb{CP}^1 && \mbox{when}&& \alpha=1 &\\
                        x\in  (0,1)&& \mbox{when}&& \alpha\neq 1&& 
                        \end{aligned}
\end{align} 
and $f$ by the surfaces
\begin{align*}
\alpha (1-x^\alpha)+X^\alpha+\left(f(x,X)\right)^\alpha=&1, 
\end{align*}
where also $X\in  \mathbb{CP}^1,$ when $\alpha=1,$ and $X\in  (0,1),$ when  $\alpha\neq 1,$
then (\ref{con_te1}) is satisfied and $S$ is a pentagon map. In this case, since for the definition of $g$ it holds $g(g(x))=x,$ mapping $S$ of Proposition \ref{prop3} reads
\begin{align} \label{S11fg}
S:(x,X;y,Y)\mapsto(u,U;v,V)=\left(g\left(\frac{g(x)}{f(v,V)}\right), X\frac{f(y,Y)}{f(v,V)};g\left(Xg(y)\right),X Y\right).
\end{align} 
Another family of the functions $g,f$ that satisfy  (\ref{con_te1}), is 
\begin{align*}
g(x)=&x, & x^\alpha+X^\alpha+\left(f(x,X)\right)^\alpha=1,
\end{align*}  
and  mapping $S$ of Proposition \ref{prop3} becomes
\begin{align}\label{S2fg}
S:(x,X;y,Y)\mapsto(u,U;v,V)=\left(\frac{x}{f(v,V)},X\frac{f(y,Y)}{f(v,V)};Xy,XY\right).
\end{align} 
Both families of maps (\ref{S11fg}), (\ref{S2fg}) which are parameterized by $\alpha,$ serve as generalizations of the families of maps presented in Proposition \ref{prop01}. Further studies on other choices of the functions $f,g,$  will be considered elsewhere.

\subsection{Parametric pentagon maps}


In this section we introduce the notion of {\em parametric pentagon maps} and as an example we show that  (\ref{S2fg}) is $(M\ddot{o}b)^2$ equivalent to a parametric pentagon map.

 Note that although parametric Yang-Baxter maps have been extensively studied \cite{Veselov:20031,ABS:YB,Veselov:07,Papageorgiou:2010,Doliwa_2014,Rizos_2013,Grahovski:2016,mikhailov2016,Kass2,KASSOTAKIS2025101094,Doikou2024,Adamopoulou_2025,Doikou11022026}, to our knowledge, there are currently no parametric pentagon maps available in the literature.
 In order to initiate the study on parametric pentagon maps,  we divide parametric maps into two classes (see Definition \ref{def1}), parametric maps that belong to class (A) and parametric maps that belong to class (B). 
 
 \begin{definition}\label{def1}
The  maps $S^{ab}: \mathcal{X}^2\times \mathcal{X}^2\rightarrow \mathcal{X}^2\times \mathcal{X}^2$ of the form
\begin{align*}
 S^{XY}: (x,X;y,Y)\mapsto & (u,X;v,Y),& \mbox{or} && S^{xy}: (x,X;y,Y)\mapsto & (x,U;y,V),& (A)
 \end{align*}
will be called parametric maps that belong to class (A). While maps $S: \mathbb{X}^2\times \mathbb{X}^2\rightarrow \mathbb{X}^2\times \mathbb{X}^2$ of the form
 \begin{align*}
 S^{xY}: (x,X;y,Y)\mapsto & (x,U;v,Y),& \mbox{or} && S^{Xy}: (x,X;y,Y)\mapsto & (u,X;y,V),& (B)
 \end{align*}
 will be called parametric maps that belong to class (B).
 \end{definition} 
 Clearly for parametric maps of class (A) the variables $X,Y$ or $x,y$ can be considered as parameters since they remain invariant under the action of the map $S$. Similarly, for parametric maps that belong to the  class (B), the variables $x,Y$ or $X,y$ are considered as parameters.  Note that in some studies on parametric Yang-Baxter maps, the parameters were not necessarily assumed to belong to the same set as the variables. Nevertheless, as it was stated in \cite{Veselov:07} the parameters could  be considered as variables. The latter viewpoint turned very useful in extending Yang-Baxter maps and integrable difference equations  into their non-abelian counterparts \cite{Doliwa_2014,Kassotakis:2:2021,Kassotakis:2022b,KASSOTAKIS2025116824}.

In order to distinguish if two parametric pentagon maps are included in the same equivalence class, we provide an equivalence relation that respects the pentagon equation.

\begin{definition}\label{def2}
Two  maps $S: \mathcal{X}^2 \times \mathcal{X}^2 \rightarrow \mathcal{X}^2 \times \mathcal{X}^2$ and $\widehat S: \mathcal{X}^2 \times \mathcal{X}^2 \rightarrow \mathcal{X}^2 \times \mathcal{X}^2$ are called $(M\ddot{o}b)^2$ equivalent if it exists a birational map $\phi: \mathcal{X}^2 \rightarrow \mathcal{X}^2$ such that $ S   (\phi\times \phi)= (\phi\times \phi) \widehat S.$
\end{definition}

\begin{prop}\label{prop1}
Let $S: \mathcal{X}^2 \times \mathcal{X}^2 \rightarrow \mathcal{X}^2 \times \mathcal{X}^2$ be a pentagon map and $\widehat S$ a $(M\ddot{o}b)^2$ equivalent map to $S$. Then $\widehat S$ is also a pentagon map.
\end{prop}

\begin{proof}
Since $\widehat S$ is  $(M\ddot{o}b)^2$ equivalent map to $S$ there exists a birational map $\phi$. Denoting $\phi_1=\phi\times Id_{\mathcal{X}^2}\times Id_{\mathcal{X}^2},$ $\phi_2=Id_{\mathcal{X}^2}\times \phi\times Id_{\mathcal{X}^2}$ and $\phi_3= Id_{\mathcal{X}^2}\times Id_{\mathcal{X}^2}\times \phi,$ we have
\begin{gather*}
\widehat S_{12} \widehat S_{13} \widehat S_{23}=  \phi_1^{-1} \phi_2^{-1} \phi_3^{-1} S_{12}  S_{13} S_{23} \phi_1  \phi_2 \phi_3
= \phi_1^{-1} \phi_2^{-1} \phi_3^{-1} S_{23}  S_{12} \phi_1 \phi_2 \phi_3=\widehat S_{23}  \widehat S_{12},
\end{gather*}
where we used the hypothesis that $S$ is a pentagon map.
\end{proof}
Note that the parametric maps that belong to the class $(B)$ are $(M\ddot{o}b)^2$ equivalent where the bijection $\phi: \mathcal{X}^2 \rightarrow \mathcal{X}^2,$ is the partial transposition $\phi:(x,X)\mapsto (X,x).$   The same holds true for parametric maps that belong to the class $(A).$

The parametric pentagon maps we obtain in the following Proposition belong to the class $(B).$

\begin{prop}
The map (\ref{S2fg})  is  $(M\ddot{o}b)^2$ equivalent to the parametric pentagon map
\begin{align*}
S^{xY}:(x,X;y,Y)\mapsto\left(x,\frac{Xy}{\left(1+y^\alpha+Y^\alpha\right)^{\frac{1}{\alpha}}};\frac{\left((1+x^\alpha)(1+y^\alpha+Y^\alpha)+
X^\alpha y^\alpha\right)^{\frac{1}{\alpha}}}{X},Y\right).
\end{align*}
\end{prop}

\begin{proof}
The pentagon map (\ref{S2fg}), reads
\begin{align*}
S:(x,X;y,Y)\mapsto(u,U;v,V)=\left(\frac{x}{f(v,V)},X\frac{f(y,Y)}{f(v,V)};Xy,XY\right),
\end{align*} 
where $f$  is defined by the family of surfaces 
\begin{align*}
 x^\alpha+X^\alpha+\left(f(x,X)\right)^\alpha=&1, & \alpha>0.
\end{align*}  
From the definition of $S,$ we observe that $\frac{V}{v}=\frac{Y}{y},$ so
\begin{align*}
J(y,Y):=&\frac{Y}{y},
\end{align*}
 is an invariant function of $S$ that involves only $y$ and $Y$.  Also, by eliminating $f(v,V)$ from the first two components of $S$ and by using (\ref{con_te1}) that holds, we obtain 
$
  \frac{f(u,U)}{u}= \frac{f(x,X)}{x}
$   
so
\begin{align*}
K(x,X):=&\frac{f(x,X)}{x},
\end{align*}
 is another invariant function of $S$ that involves only $x$ and $X$. The invariants $J$ and $K$ suggest to consider the map $\phi:\mathcal{X}\times\mathcal{X}\rightarrow \mathcal{X}\times\mathcal{X},$ such that
 \begin{align*}
 \phi: (x,X)\mapsto \left(\frac{f(x,X)}{x},\frac{X}{x}\right),
 \end{align*}
The inverse of $\phi$ reads
\begin{align*}
 \phi^{-1}: (x,X)\mapsto \left(\frac{1}{\left(1+x^\alpha+X^\alpha\right)^{\frac{1}{\alpha}}},\frac{X}{\left(1+x^\alpha+X^\alpha\right)^{\frac{1}{\alpha}}}\right),
 \end{align*}
and $(\phi\times \phi) S(\phi^{-1}\times \phi^{-1})$ becomes exactly $S^{xY}$. So indeed mapping (\ref{S2fg}) and $S^{xY}$ are $(M\ddot ob)^2$ equivalent and that completes the proof.

\end{proof}
Clearly the map of the Proposition above, serves as a family (parameterized by $\alpha$)  of parametric pentagon maps that belong to the class $(B).$  It defines $(M\ddot{o}b)^2$ equivalence classes of families of pentagon maps. For $\alpha\in \mathbb{N}$ and specifically for $\alpha=1,$ we have an equivalence class of   rational parametric pentagon maps.  For $\alpha=2,$ we have the trigonometric case, while for  $\alpha>2\in\mathbb{N}$ we obtain  elliptic maps of genus one and higher.


\section{Multicomponent pentagon maps revisited} \label{sec3}
Given a set equipped with a family of binary operations, one can  define families of $n-$ary operations on this set in various ways. For example, given a family of binary operations $\left<\cdot\right>_x: \mathcal{I}\times\mathcal{I}\rightarrow\mathcal{I},$ $x\in \mathcal{X},$ one  family of $n-$ary operations $\left<\cdot\right>_{x_1,\ldots,x_{n-1}}: \mathcal{I}^n\rightarrow \mathcal{I},$ $x_i\in \mathcal{X}$ is defined by
\begin{align*}
\left<a_1\ldots a_n\right>_{x_1,\ldots,x_{n-1}}:=&\left<\ldots\left<\left<\left<a_1a_2\right>_{x_{n-1}}a_3\right>_{x_{n-2}}\ldots a_{n-1}\right>_{x_2} a_n\right>_{x_1}.
\end{align*}
If now for this family of $n-$ary operations there is a map that satisfies the requirements of Theorem  \ref{theo11}, or of Corollary \ref{cor2}, then this map will be a multicomponent pentagon map. However, clearly not all multicomponent pentagon maps admit an underlying family of  of $n-$ary operations with the desired properties. For this reason, in this Section we propose  a construction that produces families of multicomponent maps from a single given map. As we shall see, when the single given map is a pentagon map, we obtain multicomponent pentagon and entwining pentagon maps.
\subsection{Two-families of multicomponent maps}
Let $R:\mathcal{X}\times \mathcal{X}\rightarrow \mathcal{X}\times \mathcal{X}$ be a  map. We consider two families of maps, the family
\begin{align*}
 {{}^{(n)}}t^{i,k,l}:\mathcal{X}^n\times \mathcal{X}^n\rightarrow & \mathcal{X}^n\times \mathcal{X}^n,& \begin{aligned}
                                                                                                    n> & 1,\\
                                                                                                    i,k=&1,\ldots,n,\\
                                                                                                    l=&0,1,\ldots,n-1,
                                                                                                  \end{aligned} 
 \end{align*}
and the family
\begin{align*}
 {{}^{(n)}}T^{i,k,l}:\mathcal{X}^n\times \mathcal{X}^n\rightarrow & \mathcal{X}^n\times \mathcal{X}^n,& \begin{aligned}
                                                                                                    n> & 1,\\
                                                                                                    i,l=&1,\ldots,n,\\
                                                                                                    k=&1,\ldots,n-l+1.
                                                                                                  \end{aligned} 
 \end{align*}
 The first family is defined as
 \begin{align}\label{t-map}
 {{}^{(n)}}t^{i,k,l}:=R_{i+k+l-1,i+n}\; R_{i+k+l-1,i+n+1}\ldots R_{i+k+l-1,i+n+k-1},
 \end{align}
where the first subscript in $R_{\cdot,}$  is considered modulo $n$, while the second subscript  $R_{,\cdot}$ is considered modulo $n$ with the agreement that we are starting with $n$ and not with $0$. 

The second family of maps is defined by
\begin{align}\label{T-map}
 {{}^{(n)}}T^{i,k,l}:={}^{(n)}t^{i,k,0}\;\;{}^{(n)}t^{i,k,1}\;\ldots\; {}^{(n)}t^{i,k,l-1},
 \end{align}
that is build from a specific composition of members of the  first family. Clearly there is ${{}^{(n)}}T^{i,k,1}\equiv {{}^{(n)}}t^{i,k,0}.$

 In the full generality we can consider the maps 
\begin{align*}
 {{}^{(n)}}\mathcal{T}^{i,{\bf k},{\bf l}}:={}^{(n)}t^{i,k_1,l_1}\;\;{}^{(n)}t^{i,k_2,l_2}\;\ldots\; {}^{(n)}t^{i,k_m,l_m},
 \end{align*}
where ${\bf k}=(k_1,k_2,\ldots, k_m),$ ${\bf l}=(l_1,l_2,\ldots l_m),$  and $k_i\in \{1,2,\ldots K\},$ $l_i\in \{1,2,\ldots L\},$ $L,M\leq n.$ Nevertheless, in this article we consider only the two families (\ref{t-map}) and (\ref{T-map}).

 In order to present  (\ref{t-map}) and (\ref{T-map}) in a compact form, we introduce the notation 
\begin{align*}
  \circ_{i=0}^{s}\;f^{i,j} &:= f^{0,j} f^{1,j}\ldots f^{s,j},\\
  \circ_{j=0}^{r} \circ_{i=0}^{s}\;f^{i,j} &:=\circ_{j=0}^{r}( f^{0,j} f^{1,j}\ldots f^{s,j}),
\end{align*}
where $r,s\in\mathbb{N}$ and $f^{i,j}$ a collection of maps.  With this annotation  (\ref{t-map}) and (\ref{T-map}) respectively read
\begin{align}\label{tt-map}
{{}^{(n)}}t^{i,k,l}:=&\circ_{j=1}^{k}R_{i+k+l-1,i+n+j-1},\\ \label{TT-map}
{{}^{(n)}}T^{i,k,l}:=&\circ_{m=0}^{l-1}\circ_{j=1}^{k}R_{i+k+m-1,i+n+j-1}.
\end{align}

As an example we present the aforementioned  maps for $n=2.$ When $n=2$, there is $i,k=1,2,$ $l=0,1$ and the maps  ${{}^{(2)}}t^{i,k,0}$ explicitly read
  \begin{align*}
  {{}^{(2)}}t^{1,1,0}=&R_{1,3}, & {{}^{(2)}}t^{2,1,0}=&R_{2,4},&
  {{}^{(2)}}t^{1,2,0}=&R_{2,3}R_{2,4}, & {{}^{(2)}}t^{2,2,0}=&R_{1,4}R_{1,3}, 
  \end{align*}
   while the maps  ${{}^{(2)}}t^{i,k,1}$ read
  \begin{align*}
  {{}^{(2)}}t^{1,1,1}=&R_{2,3}, & {{}^{(2)}}t^{2,1,1}=&R_{1,4},&
  {{}^{(2)}}t^{1,2,1}=&R_{1,3}R_{1,4}, & {{}^{(2)}}t^{2,2,1}=&R_{2,4}R_{2,3}.
\end{align*}  
For the maps ${{}^{(2)}}T^{i,k,l},$ we have $i,l=1,2,$ $k=1,\ldots,3-l,$ and they explicitly read
\begin{align*}
  {{}^{(2)}}T^{1,1,1}= {{}^{(2)}}t^{1,1,0}=&R_{1,3},& {{}^{(2)}}T^{2,1,1}= {{}^{(2)}}t^{2,1,0}=&R_{2,4}, \\
   {{}^{(2)}}T^{1,2,1}= {{}^{(2)}}t^{1,2,0}=&R_{2,3}R_{2,4},& {{}^{(2)}}T^{2,2,1}= {{}^{(2)}}t^{2,2,0}=&R_{1,4}R_{1,3},\\
  {{}^{(2)}}T^{1,1,2}= {{}^{(2)}}t^{1,1,0}\; {{}^{(2)}}t^{1,1,1}=&R_{1,3}R_{2,3},& {{}^{(2)}}T^{2,1,2}= {{}^{(2)}}t^{2,1,0}\; {{}^{(2)}}t^{2,1,1}=&R_{2,4}R_{1,4}.
  \end{align*}
  
  When $n=3,$ we present  the maps ${{}^{(3)}}T^{i,k,l},$  where $i,l=1,2,3,$ $k=1,\ldots,4-l,$ only for $i=1.$ These maps explicitly read 
  \begin{align*}
  {{}^{(3)}}T^{1,1,1}=&R_{1,4},& {{}^{(3)}}T^{1,2,1}=&R_{2,4}R_{2,5}, & {{}^{(3)}}T^{1,3,1}=&R_{3,4}R_{3,5}R_{3,6},\\
  {{}^{(3)}}T^{1,1,2}=&R_{1,4}R_{2,4},& {{}^{(3)}}T^{1,2,2}=&R_{2,4}R_{2,5}R_{3,4}R_{3,5}, & &\\
  {{}^{(3)}}T^{1,1,3}=&R_{1,4}R_{2,4}R_{3,4}.& & & &
    \end{align*}



\subsection{Transfer like pentagon maps}

Assume now that $R:\mathcal{X}\times \mathcal{X}\rightarrow \mathcal{X}\times \mathcal{X}$ is a pentagon map. As in the previous Section, in order to distinguish  pentagon maps  we introduce the following equivalence relation that respects the pentagon equation.

\begin{definition}\label{equi2}
Two  maps $S: \mathcal{X}^n \times \mathcal{X}^n \rightarrow \mathcal{X}^n \times \mathcal{X}^n$ and $\widehat S: \mathcal{X}^n \times \mathcal{X}^n \rightarrow \mathcal{X}^n \times \mathcal{X}^n$ are called $(M\ddot{o}b)^n$ equivalent if it exists a birational map $\phi: \mathcal{X}^n \rightarrow \mathcal{X}^n$ such that $ S\circ  (\phi\times \phi)= (\phi\times \phi)\circ\widehat S.$
\end{definition}    
  
\begin{remark} \label{rem1}
 The families of maps ${{}^{(n)}}t^{i+1,k,l}$ and ${{}^{(n)}}T^{i+1,k,l},$ are families  of  equivalence classes of maps, under the equivalence relation of Definition \ref{equi2}. 
Indeed, if we consider the birational map $\phi:\mathcal{X}^n\rightarrow \mathcal{X}^n$ defined as $\phi: (x_1,x_2,\ldots,x_n)\mapsto (x_2,x_3,\ldots,x_1),$ there is $(\phi\times\phi)^{-1}\; {{}^{(n)}}t^{i,k,l} (\phi\times\phi)={{}^{(n)}}t^{i+1,k,l},$ and similarly for ${{}^{(n)}}T^{i,k,l}.$ 
\end{remark}

To the families of maps ${{}^{(n)}}t^{i+1,k,l}$ and ${{}^{(n)}}T^{i+1,k,l},$ we associate the maps ${{}^{(n)}}t^{i,k,l}_{{\bf st}},{{}^{(n)}}T^{i,k,l}_{{\bf st}}:\mathcal{X}^n \times \mathcal{X}^n\times \mathcal{X}^n\rightarrow \mathcal{X}^n \times \mathcal{X}^n\times \mathcal{X}^n,$ ${\bf s}< {\bf t}\in \{{\bf 1},{\bf 2},{\bf 3}\},$ where   ${\bf 1}:=(1,2,\ldots,n),$ ${\bf 2}:=(n+1,n+2,\ldots,2n),$ ${\bf 3}:=(2n+1,2n+2,\ldots,3n),$ which are defined as  
\begin{align*}
{{}^{(n)}}t^{i,k,l}_{{\bf 12}}:=&{{}^{(n)}}t^{i,k,l}\times id_{\mathcal{X}^n},&
{{}^{(n)}}t^{i,k,l}_{{\bf 23}}:=&id_{\mathcal{X}^n}\times {{}^{(n)}}t^{i,k,l},&
{{}^{(n)}}t^{i,k,l}_{{\bf 13}}:=&(\tau\times id_{\mathcal{X}^n}) {{}^{(n)}}t^{i,k,l}_{{\bf 23}} (\tau\times id_{\mathcal{X}^n}),
\end{align*}
and 
\begin{align*}
{{}^{(n)}}T^{i,k,l}_{{\bf 12}}:=&{{}^{(n)}}T^{i,k,l}\times id_{\mathcal{X}^n},&
{{}^{(n)}}T^{i,k,l}_{{\bf 23}}:=&id_{\mathcal{X}^n}\times {{}^{(n)}}T^{i,k,l},&
{{}^{(n)}}T^{i,k,l}_{{\bf 13}}:=&(\tau\times id_{\mathcal{X}^n}) {{}^{(n)}}T^{i,k,l}_{{\bf 23}} (\tau\times id_{\mathcal{X}^n}),
\end{align*}
where $\tau:\mathcal{X}^n\times \mathcal{X}^n\rightarrow \mathcal{X}^n\times \mathcal{X}^n,$ the transposition map that is $\tau:({\bf x},{\bf y})\mapsto ({\bf y},{\bf x}).$ In terms of $R:\mathcal{X}\times \mathcal{X}\rightarrow \mathcal{X}\times \mathcal{X},$ they respectively read:
\begin{align}\label{tij}
\begin{aligned}
{{}^{(n)}}t^{i,k,l}_{{\bf 12}}=&\circ_{j=1}^{k}R_{i+k+l-1,i+n+j-1},\\
{{}^{(n)}}t^{i,k,l}_{{\bf 13}}=&\circ_{j=1}^{k}R_{i+k+l-1,i+2n+j-1},\\
{{}^{(n)}}t^{i,k,l}_{{\bf 23}}=&\circ_{j=1}^{k}R_{i+n+k+l-1,i+2n+j-1},
\end{aligned}
\end{align}
and
\begin{align}\label{T12}
{{}^{(n)}}T^{i,k,l}_{{\bf 12}}=&\circ_{m=0}^{l-1}{{}^{(n)}}t^{i,k,m}_{{\bf 12}}=\circ_{m=0}^{l-1}\circ_{j=1}^{k}R_{i+k+m-1,i+n+j-1},\\ \label{T13}
{{}^{(n)}}T^{i,k,l}_{{\bf 13}}=&\circ_{m=0}^{l-1}{{}^{(n)}}t^{i,k,m}_{{\bf 13}}=\circ_{m=0}^{l-1}\circ_{j=1}^{k}R_{i+k+m-1,i+2n+j-1},\\ \label{T23}
{{}^{(n)}}T^{i,k,l}_{{\bf 23}}=&\circ_{m=0}^{l-1}{{}^{(n)}}t^{i,k,m}_{{\bf 23}}=\circ_{m=0}^{l-1}\circ_{j=1}^{k}R_{i+n+k+m-1,i+2n+j-1}.
\end{align}
In the formulas above, the bolted subscript ${\bf 1}$ is considered modulo $n$, the  subscript ${\bf 2}$ is considered modulo $n$ with the agreement that we are starting with $n$ and not with $0$, and the subscript ${\bf 3}$ is considered modulo $n$ where we are starting with $2n.$ 

 We have the following Lemma

 \begin{lemma}\label{lemma1}
 Let $R:\mathcal{X}\times \mathcal{X}\rightarrow \mathcal{X}\times \mathcal{X}$ be a pentagon map. \\
 A. The families of maps (\ref{tt-map})
 \begin{align*}
 {{}^{(n)}}t^{i,k,l}:\mathcal{X}^n\times \mathcal{X}^n\rightarrow & \mathcal{X}^n\times \mathcal{X}^n,& \begin{aligned}
                                                                                                    n> & 1,\\
                                                                                                    i,k=&1,\ldots,n,\\
                                                                                                    l=&0,1,\ldots,n-1,
                                                                                                  \end{aligned} 
 \end{align*}
 satisfy:
 \begin{flalign*}
 &(I_A)& [{{}^{(n)}}t^{i,k,m}_{{\bf 12}},\;{{}^{(n)}}t^{i,k,m'}_{{\bf 23}}]=&0,&\begin{aligned} &  m\in \{0,1,\ldots, l-1\},\\
                                                                                        &    m'\in \{1,\ldots, l-1\},
                                                                                       \end{aligned}\\
&(I_B)& [{{}^{(n)}}t^{i,k,m}_{{\bf 12}},\;{{}^{(n)}}t^{i,k,m'}_{{\bf 13}}]=&0,&m\neq m',                                                                                       
 \end{flalign*}
 \begin{flalign*}
 &(II)& \;{{}^{(n)}}t^{i,k,0}_{{\bf 23}}\;R_{i+k+q-1,i+n+k-1}=& R_{i+k+q-1,i+n+k-1} \;{{}^{(n)}}t^{i,k,q}_{{\bf 13}}\; \;{{}^{(n)}}t^{i,k,0}_{{\bf 23}}, &q=&0,1,\ldots,l-1.
\end{flalign*}

 \begin{flalign*}
 &(III)&        {{}^{(n)}}t^{i,k,0}_{{\bf 23}}\; {{}^{(n)}}t^{i,k,q}_{{\bf 12}}=& {{}^{(n)}}t^{i,k,q}_{{\bf 12}}\; {{}^{(n)}}t^{i,k,q}_{{\bf 13}} \;{{}^{(n)}}t^{i,k,0}_{{\bf 23}},& q=0,1,\ldots, l-1.
       \end{flalign*}
  B. The families  ${{}^{(n)}}t^{i,k,0},$ together with  the families of maps (\ref{TT-map})
 \begin{align*}
 {{}^{(n)}}T^{i,k,l}:\mathcal{X}^n\times \mathcal{X}^n\rightarrow & \mathcal{X}^n\times \mathcal{X}^n,& \begin{aligned}
                                                                                                    n> & 1,\\
                                                                                                    i,l=&1,\ldots,n,\\
                                                                                                    k=&1,\ldots,n-l+1.
                                                                                                  \end{aligned} 
 \end{align*}
 satisfy:
 \begin{flalign*}
  &(IV)&       {{}^{(n)}}t^{i,k,0}_{{\bf 23}}\; {{}^{(n)}}T^{i,k,l}_{{\bf 12}}=& {{}^{(n)}}T^{i,k,l}_{{\bf 12}}\; {{}^{(n)}}T^{i,k,l}_{{\bf 13}} \;{{}^{(n)}}t^{i,k,0}_{{\bf 23}},& 
       \end{flalign*}

 \end{lemma}
 \begin{proof}
 The proof of the Lemma is given in Appendix \ref{app1}.
 \end{proof}
 Note that due to items $(I_A), (I_B)$ and $(II)$ of  Lemma \ref{lemma1},  the family of maps ${{}^{(n)}}t^{i,k,l}$ resemble  the {\em transfer maps} associated with Yang-Baxter maps \cite{Veselov:20031},  so we refer to them as {\em transfer like maps} associated with the pentagon map $R$. Item $(II)$ says that mapping $R$ and specific members of the families ${{}^{(n)}}t^{i,k,l}$ satisfy the entwining pentagon equation, while in item $(III)$ we have that 
 ${{}^{(n)}}t^{i,k,0}$ entwines with ${{}^{(n)}}t^{i,k,q},$ $q=0,\ldots, l-1.$ Furthermore,  item $(IV)$ says that ${{}^{(n)}}t^{i,k,0}$ and  ${{}^{(n)}}T^{i,k,l},$  satisfy the entwining pentagon equation.
 
The following  Theorem   allow us to construct families of multicomponent pentagon maps from a given one-component  pentagon  map.
 \begin{theorem}\label{theo2}
 Let $R:\mathcal{X}\times \mathcal{X}\rightarrow \mathcal{X}\times \mathcal{X}$ be a pentagon map. Then the   maps  (\ref{TT-map})
 \begin{align*}
 {{}^{(n)}}T^{i,k,l}:\mathcal{X}^n\times \mathcal{X}^n\rightarrow & \mathcal{X}^n\times \mathcal{X}^n,& \begin{aligned}
                                                                                                    n> & 1,\\
                                                                                                    i,l=&1,\ldots,n,\\
                                                                                                    k=&1,\ldots,n-l+1.
                                                                                                  \end{aligned} 
 \end{align*}
are families of pentagon maps, that is they satisfy
 \begin{align}\label{the1_2}
         {{}^{(n)}}T^{i,k,l}_{{\bf 23}}\; {{}^{(n)}}T^{i,k,l}_{{\bf 12}}=& {{}^{(n)}}T^{i,k,l}_{{\bf 12}}\; {{}^{(n)}}T^{i,k,l}_{{\bf 13}} \;{{}^{(n)}}T^{i,k,l}_{{\bf 23}}.
       \end{align}
 \end{theorem}
 \begin{proof}
 From the definitions (\ref{T12}) and (\ref{T23}) of ${{}^{(n)}}T^{i,k,l}_{{\bf 12}}$ and ${{}^{(n)}}T^{i,k,l}_{{\bf 23}},$ there is 
 \begin{align*}
   \begin{aligned}
   {{}^{(n)}}T^{i,k,l}_{{\bf 23}}\; {{}^{(n)}}T^{i,k,l}_{{\bf 12}}=& {{}^{(n)}}t^{i,k,0}_{{\bf 23}}\; {{}^{(n)}}t^{i,k,1}_{{\bf 23}}\ldots {{}^{(n)}}t^{i,k,l-1}_{{\bf 23}}\; {{}^{(n)}}t^{i,k,0}_{{\bf 12}}\; {{}^{(n)}}t^{i,k,1}_{{\bf 12}}\ldots {{}^{(n)}}t^{i,k,l-1}_{{\bf 12}}\\
   =&{{}^{(n)}}t^{i,k,0}_{{\bf 23}}\; {{}^{(n)}}t^{i,k,0}_{{\bf 12}}\; \; {{}^{(n)}}t^{i,k,1}_{{\bf 12}}\ldots {{}^{(n)}}t^{i,k,l-1}_{{\bf 12}}\; {{}^{(n)}}t^{i,k,1}_{{\bf 23}}\; {{}^{(n)}}t^{i,k,2}_{{\bf 23}}\ldots {{}^{(n)}}t^{i,k,l-1}_{{\bf 23}}\\
   =&{{}^{(n)}}t^{i,k,0}_{{\bf 23}}\; {{}^{(n)}}T^{i,k,l}_{{\bf 12}}\; {{}^{(n)}}t^{i,k,1}_{{\bf 23}}\; {{}^{(n)}}t^{i,k,2}_{{\bf 23}}\ldots {{}^{(n)}}t^{i,k,l-1}_{{\bf 23}},
   \end{aligned}
 \end{align*}
 where we used item $(I_A)$ of Lemma \ref{lemma1}. Now by using item $(IV)$ of the same Lemma we obtain
 \begin{align*}
   \begin{aligned}
   {{}^{(n)}}T^{i,k,l}_{{\bf 23}}\; {{}^{(n)}}T^{i,k,l}_{{\bf 12}}=& 
   {{}^{(n)}}t^{i,k,0}_{{\bf 23}}\; {{}^{(n)}}T^{i,k,l}_{{\bf 12}}\; {{}^{(n)}}t^{i,k,1}_{{\bf 23}}\; {{}^{(n)}}t^{i,k,2}_{{\bf 23}}\ldots {{}^{(n)}}t^{i,k,l-1}_{{\bf 23}}\\
   =& {{}^{(n)}}T^{i,k,l}_{{\bf 12}}\; {{}^{(n)}}T^{i,k,l}_{{\bf 13}} \;{{}^{(n)}}t^{i,k,0}_{{\bf 23}}\; {{}^{(n)}}t^{i,k,1}_{{\bf 23}}\; \ldots \; {{}^{(n)}}t^{i,k,l-1}_{{\bf 23}}\\
   =& {{}^{(n)}}T^{i,k,l}_{{\bf 12}}\; {{}^{(n)}}T^{i,k,l}_{{\bf 13}} \; {{}^{(n)}}T^{i,k,l}_{{\bf 12}},   
    \end{aligned}
 \end{align*}
 and that completes the proof. 
 \end{proof}

\section*{Acknowledgements}
This paper has been financed by the funding programme ``MEDICUS", of the University of Patras.

We would like to  thank the organisers of the conference {\em Discrete Integrable Systems}, (6-10  April 2026, TSIMF, Sanya, China) for the invitation and  the warm hospitality. This article was finalized during our stay.

\appendix

\section{Proof of Lemma \ref{lemma1}}\label{app1}

Let us first prove part A. of the Lemma. In order to prove item $(I_A)$, we have to show that only when $m'\neq 0,$ the maps ${{}^{(n)}}t^{i,k,m}_{{\bf 12}},$ $m\in \{0,1,\ldots, l-1\},$ and the maps $\;{{}^{(n)}}t^{i,k,m'}_{{\bf 23}},$  $m'\in \{1,\ldots, l-1\},$ act non-trivially on different  factors  of  $\mathcal{X}^n\times \mathcal{X}^n\times \mathcal{X}^n.$ From (\ref{tij}), we have \begin{align*}
 {{}^{(n)}}t^{i,k,m}_{{\bf 12}}=&\circ_{j=1}^{k}R_{i+k+m-1,i+n+j-1},&
{{}^{(n)}}t^{i,k,m'}_{{\bf 23}}=&\circ_{j=1}^{k}R_{i+n+k+m'-1,i+2n+j-1},
\end{align*}
which in order to act non-trivially on different  factors  of  $\mathcal{X}^n\times \mathcal{X}^n\times \mathcal{X}^n,$ it should be
\begin{align*}
i+n+j-1\neq & i+n+k+m'-1,& \forall j\in & \{1,\ldots,k\},
\end{align*}
or equivalently
\begin{align*}
j\neq & k+m',& \forall j\in & \{1,\ldots,k\}, 
\end{align*}
that holds only for $m'\neq 0.$ In a similar way we can show that $(I_B)$ holds.

To prove item $(II)$, note that there is
\begin{align}\label{eqp1} 
\begin{aligned}
\;{{}^{(n)}}t^{i,k,0}_{{\bf 23}}\;R_{i+k+q-1,i+n+k-1}=& (\circ_{j=1}^{k}R_{i+n+k-1,i+2n+j-1}) \;R_{i+k+q-1,i+n+k-1}\\
&= R_{i+k+q-1,i+n+k-1} \circ_{j=1}^{k}R_{i+k+q-1,i+2n+j-1} R_{i+n+k-1,i+2n+j-1},
\end{aligned}
\end{align}
where we used the definition of ${{}^{(n)}}t^{i,k,0}_{{\bf 23}}$ from (\ref{tij}) and the hypothesis that $R$ is a pentagon map,  so for any of the triples  $(i+k+q-1,i+n+k-1,i+2n+j-1),$ of subscripts, with $i=1,2,\ldots,n,$ $j=1,2,\ldots,k,$                                                                                                         $q=0,1,\ldots,l-1,$  it holds 
\begin{align*}
R_{i+n+k-1,i+2n+j-1} R_{i+k+q-1,i+n+k-1}=&R_{i+k+q-1,i+n+k-1} R_{i+k+q-1,i+2n+j-1}R_{i+n+k-1,i+2n+j-1}. 
\end{align*}
 Since for $q\neq n$ the pairs $(i+k+q-1,i+2n+j),(i+n+k-1,i+2n+j-1)$ are disjoint, there is
 \begin{align*}
 [R_{i+k+q-1,i+2n+j},R_{i+n+k-1,i+2n+j-1}]=&0, q\neq n,
 \end{align*}
 and the product of composition of maps $\circ_{j=1}^{k}R_{i+k+q-1,i+2n+j-1} R_{i+n+k-1,i+2n+j-1}$ that appears in (\ref{eqp1}), can be rewritten as
 \begin{align}\label{eqp2}
 \begin{aligned}
    \underbrace{R_{i+k+q-1,i+2n} \ldots R_{i+k+q-1,i+2n+k-1}}_{\circ_{j=1}^{k}R_{i+k+q-1,i+2n+j-1}}\underbrace{R_{i+n+k-1,i+2n} \ldots R_{i+n+k-1,i+2n+k-1}}_{\circ_{j=1}^{k}R_{i+n+k-1,i+2n+j-1}}& \\
    &= {{}^{(n)}}t^{i,k,q}_{{\bf 13}}\; {{}^{(n)}}t^{i,k,0}_{{\bf 23}}.
   \end{aligned}
 \end{align}  
  Replacing (\ref{eqp2}) in (\ref{eqp1}) we obtain item $(II)$.

  Now we prove item $(III)$. We have
  \begin{align*} 
\begin{aligned}
\;{{}^{(n)}}t^{i,k,0}_{{\bf 23}}\;\;{{}^{(n)}}t^{i,k,q}_{{\bf 12}}=& (\circ_{j=1}^{k}R_{i+n+k-1,i+2n+j-1}) \;\circ_{j=1}^{k}R_{i+q+k-1,i+n+j-1}\\
&= (\circ_{j=1}^{k}R_{i+n+k-1,i+2n+j-1}) \;(\circ_{j=1}^{k-1}R_{i+q+k-1,i+n+j-1})\;R_{i+q+k-1,i+n+k-1}.
\end{aligned}
\end{align*}
Since $q<n$ and $k>0$, there is $[\circ_{j=1}^{k}R_{i+n+k-1,i+2n+j-1}, \;\circ_{j=1}^{k-1}R_{i+q+k-1,i+n+j-1}]=0,$ so the formula above reads
 \begin{align*}
  \begin{aligned}
\;{{}^{(n)}}t^{i,k,0}_{{\bf 23}}\;\;{{}^{(n)}}t^{i,k,q}_{{\bf 12}}=& (\circ_{j=1}^{k-1}R_{i+q+k-1,i+n+j-1})\; (\circ_{j=1}^{k}R_{i+n+k-1,i+2n+j-1}) \;R_{i+q+k-1,i+n+k-1}\\
=&\circ_{j=1}^{k-1}R_{i+q+k-1,i+n+j-1}\;\underbrace{{{}^{(n)}}t^{i,k,0}_{{\bf 23}}\;R_{i+q+k-1,i+n+k-1}}_{\begin{aligned}=&R_{i+q+k-1,i+n+k-1} \;{{}^{(n)}}t^{i,k,q}_{{\bf 13}}\; \;{{}^{(n)}}t^{i,k,0}_{{\bf 23}} &\\ &\mbox{due to item  (II)}\end{aligned}}\\
=&\circ_{j=1}^{k}R_{i+q+k-1,i+n+j-1}\;\;{{}^{(n)}}t^{i,k,q}_{{\bf 13}}\; \;{{}^{(n)}}t^{i,k,0}_{{\bf 23}}
= {{}^{(n)}}t^{i,k,q}_{{\bf 12}}\;\;{{}^{(n)}}t^{i,k,q}_{{\bf 13}}\; \;{{}^{(n)}}t^{i,k,0}_{{\bf 23}},
\end{aligned}
\end{align*}
and that completes the proof.

For item $(IV)$ we have
\begin{align*} 
\begin{aligned}
\;{{}^{(n)}}t^{i,k,0}_{{\bf 23}}\;\;{{}^{(n)}}T^{i,k,l}_{{\bf 12}}=&\underbrace{{{}^{(n)}}t^{i,k,0}_{{\bf 23}}\;\;{{}^{(n)}}t^{i,k,0}_{{\bf 12}}\;{{}^{(n)}}}_{\begin{aligned}=&{{}^{(n)}}t^{i,k,0}_{{\bf 12}}\; {{}^{(n)}}t^{i,k,0}_{{\bf 13}} \;{{}^{(n)}}t^{i,k,0}_{{\bf 23}}\\
& \mbox{due to item (III)}\end{aligned}}\;t^{i,k,1}_{{\bf 12}}\ldots {{}^{(n)}}t^{i,k,l-1}_{{\bf 12}} \\
&= {{}^{(n)}}t^{i,k,0}_{{\bf 12}}\; {{}^{(n)}}t^{i,k,0}_{{\bf 13}} \;\underbrace{{{}^{(n)}}t^{i,k,0}_{{\bf 23}}\; t^{i,k,1}_{{\bf 12}}}_{\begin{aligned}=&{{}^{(n)}}t^{i,k,1}_{{\bf 12}}\; {{}^{(n)}}t^{i,k,1}_{{\bf 13}} \;{{}^{(n)}}t^{i,k,0}_{{\bf 23}}\\
& \mbox{due to item (III)}\end{aligned}}\ldots {{}^{(n)}}t^{i,k,l-1}_{{\bf 12}}.
\end{aligned}
\end{align*}
Repeating the process above $l-2$ times, we obtain
\begin{align*} 
\begin{aligned}
\;{{}^{(n)}}t^{i,k,0}_{{\bf 23}}\;\;{{}^{(n)}}T^{i,k,l}_{{\bf 12}}=& {{}^{(n)}}t^{i,k,0}_{{\bf 12}}\left(\circ_{j=0}^{l-2}\;  {{}^{(n)}}t^{i,k,j}_{{\bf 13}}\;  {{}^{(n)}}t^{i,k,j+1}_{{\bf 12}}\right)\;  {{}^{(n)}}t^{i,k,l-1}_{{\bf 13}}\;  {{}^{(n)}}t^{i,k,0}_{{\bf 23}}.
\end{aligned}
\end{align*}
Using item $(I_B)$ the formula above becomes
\begin{align*} 
\begin{aligned}
\;{{}^{(n)}}t^{i,k,0}_{{\bf 23}}\;\;{{}^{(n)}}T^{i,k,l}_{{\bf 12}}=& {{}^{(n)}}t^{i,k,0}_{{\bf 12}}\left(\circ_{j=0}^{l-2}\;   {{}^{(n)}}t^{i,k,j+1}_{{\bf 12}}\;  {{}^{(n)}}t^{i,k,j}_{{\bf 13}}\right)\;  {{}^{(n)}}t^{i,k,l-1}_{{\bf 13}}\;  {{}^{(n)}}t^{i,k,0}_{{\bf 23}}\\
&=\underbrace{{{}^{(n)}}t^{i,k,0}_{{\bf 12}}\; {{}^{(n)}}t^{i,k,1}_{{\bf 12}}\ldots {{}^{(n)}}t^{i,k,l-1}_{{\bf 12}}}_{={{}^{(n)}}T^{i,k,l}_{{\bf 12}}}\;
\underbrace{{{}^{(n)}}t^{i,k,0}_{{\bf 13}}\; {{}^{(n)}}t^{i,k,1}_{{\bf 13}}\ldots {{}^{(n)}}t^{i,k,l-1}_{{\bf 13}}}_{={{}^{(n)}}T^{i,k,l}_{{\bf 13}}}\;{{}^{(n)}}t^{i,k,0}_{{\bf 23}}\\
=&{{}^{(n)}}T^{i,k,l}_{{\bf 12}}\; {{}^{(n)}}T^{i,k,l}_{{\bf 13}}\; \;{{}^{(n)}}t^{i,k,0}_{{\bf 23}},
\end{aligned}
\end{align*}
that is exactly item (IV) of the Lemma.


\end{document}